\documentclass[11pt]{article}
\usepackage[margin=1in]{geometry}
\usepackage[latin2]{inputenc}
\usepackage[english]{babel}
\usepackage[cmex10]{amsmath}
\usepackage[all=normal,paragraphs=tight,bibliography=tight]{savetrees}
\usepackage{todonotes}
\usepackage{amsthm}
\usepackage{amssymb}
\usepackage{microtype}
\usepackage{xspace}
\usepackage{comment}
\usepackage{letltxmacro}
\usepackage{comment}
\usepackage{enumerate}
\usepackage{enumitem}
\usepackage{tikz}
\usetikzlibrary{calc}
\usepackage{url}

\usepackage{graphicx}
\usepackage{caption}
\usepackage{subcaption}
\usepackage{float}

\newtheorem{theorem}{Theorem}[section]
\newtheorem{lemma}[theorem]{Lemma}
\newtheorem{myclaim}[theorem]{Claim}
\newtheorem{corollary}[theorem]{Corollary}

\theoremstyle{definition}

\newcommand{\iflncs}[2]{\ifthenelse{\equal{\uselncs}{yes}}{#1}{#2}}
\newcommand{\onlylncs}[1]{\iflncs{#1}{}}
\newcommand{\onlyfull}[1]{\iflncs{}{#1}}
\newcommand{\uselncs}{no}

\newcommand{\maybeqed}{\onlylncs{\qed}}

\newcommand{\cH}{{\mathcal{H}}}
\newcommand{\cB}{{\mathcal{B}}}
\newcommand{\cV}{{\mathcal{V}}}

\newcommand{\Oh}{\ensuremath{\mathcal{O}}}

\def\cqedsymbol{\ifmmode$\lrcorner$\else{\unskip\nobreak\hfil
\penalty50\hskip1em\null\nobreak\hfil$\lrcorner$
\parfillskip=0pt\finalhyphendemerits=0\endgraf}\fi} 

\newcommand{\cqed}{\renewcommand{\qed}{\cqedsymbol}}

\newcommand{\executeiffilenewer}[3]{%
\ifnum\pdfstrcmp{\pdffilemoddate{#1}}%
{\pdffilemoddate{#2}}>0%
{\immediate\write18{#3}}\fi%
} 

\definecolor{Eblue}{RGB}{90,150,190}
\tikzset{%
	v/.style={circle,draw=black!75,inner sep=0pt,minimum size=18pt},
	p/.style={circle,fill=black!75,inner sep=2pt},
	s/.style={circle,fill=black!75,inner sep=0pt},
	E/.style={ultra thick,draw=Eblue},
	non/.style={dashed},
	bag/.style={fill=gray!18!white!78!yellow},
	bagl/.style={fill=gray!10!white},
	Ebag/.style={ultra thick, black!56!yellow!40!gray},
}

\newcommand{\clawdiamond}{\ensuremath{\{\mbox{claw},\mbox{diamond}\}}\xspace}

\newcommand{\cdedgedeletion}{{\sc{\{claw,diamond\}-free Edge Deletion}}\xspace}
\newcommand{\cdedgedeletionfull}{\cdedgedeletion}

\title{Polynomial kernelization for removing induced claws and diamonds\thanks{The research was supported by Polish National Science Centre grants DEC-2013/11/D/ST6/03073 (Micha\l{} Pilipczuk and Marcin Wrochna) and DEC-2012/05/D/ST6/03214 (Marek Cygan and Marcin Pilipczuk).
  Micha\l{} Pilipczuk is currently holding a post-doc position at Warsaw Center of Mathematics and Computer Science.}}
\author{
  Marek Cygan\thanks{Institute of Informatics, University of Warsaw, Poland, \texttt{\{cygan,michal.pilipczuk,m.wrochna\}@mimuw.edu.pl}.}
  \and
  Marcin Pilipczuk\thanks{Department of Computer Science, University of Warwick, UK, \texttt{m.pilipczuk@dcs.warwick.ac.uk}.}
  \and
  Micha\l{} Pilipczuk$^{\dagger}$
  \and 
  Erik Jan van Leeuwen\thanks{Max-Planck Institut f\"{u}r Informatik, Saarbr\"{u}cken, Germany, \texttt{erikjan@mpi-inf.mpg.de}.}
  \and 
  Marcin Wrochna$^{\dagger}$
  }

\date{}

\begin{document}

\maketitle

\begin{abstract}
A graph is called \clawdiamond-free if it contains neither a claw (a $K_{1,3}$) nor a diamond (a $K_4$ with an edge removed) as an induced subgraph. Equivalently, \clawdiamond-free graphs can be characterized as line graphs of triangle-free graphs, or as \emph{linear dominoes}, i.e., graphs in which every vertex is in at most two maximal cliques and every edge is in exactly one maximal clique. 

In this paper we consider the parameterized complexity of the {\sc{\{claw,diamond\}-free Edge Deletion}} problem, where given a graph $G$ and a parameter $k$, the question is whether one can remove at most $k$ edges from $G$ to obtain a \clawdiamond-free graph. Our main result is that this problem admits a polynomial kernel. We complement this finding by proving that, even on instances with maximum degree $6$, the problem is NP-complete and cannot be solved in time $2^{o(k)}\cdot |V(G)|^{\Oh(1)}$ unless the Exponential Time Hypothesis fails.

\end{abstract}

\section{Introduction}\label{sec_intro}
Graph modification problems form a wide class of problems, where one is asked to alter a given graph using a limited number of modifications in order to achieve a certain target property, for instance the non-existence of some forbidden induced structures. Depending on the allowed types of modification and the choice of the target property, one can consider a full variety of problems. Well-studied problems that can be expressed in the graph modification paradigm are {\sc{Vertex Cover}}, {\sc{Feedback Vertex Set}}, and {\sc{Cluster Editing}}, among others.

It is natural to consider graph modification problems from the parameterized perspective, since they have an innate parameter: the number of allowed modifications, which is expected to be small in applications. As far as the set of allowed modifications is concerned, the most widely studied variants are vertex deletion problems (allowing only removing vertices), edge deletion problems (only removing edges), completion problems (only adding edges), and editing problems (both adding and removing edges). It is very easy to see that as long as the target property can be expressed as the non-existence of induced subgraphs from some finite, fixed list of forbidden subgraphs $\mathcal{F}$ (in other words, belonging to the class of $\mathcal{F}$-free graphs), then all the four variants can be solved in time $c^k\cdot |V(G)|^{\Oh(1)}$ via a straightforward branching strategy, where the constant $c$ depends on $\mathcal{F}$ only. This observation was first pronounced by Cai~\cite{cai1996fixed}.

From the perspective of kernelization, again whenever the property is characterized by a finite list of forbidden induced subgraphs, then a standard application of the sunflower lemma gives a polynomial kernel for the vertex deletion variant. The same observation, however, does not carry over to the edge modification problems. The reason is that altering one edge can create new obstacles from $\mathcal{F}$, which need to be dealt with despite not being contained in the original graph $G$. Indeed, Kratsch and Wahlstr\"om~\cite{kratsch2009two} have shown a simple graph $H$ on $7$ vertices such that the edge deletion problem for the property of being $H$-free does not admit a polynomial kernel unless $\textrm{NP}\subseteq \textrm{coNP}/\textrm{poly}$. Later, the same conclusion was proved by Guillemot et al.~\cite{guillemot2013non} for $H$ being a long enough path or cycle.

This line of study was continued by Cai and Cai~\cite{cai2013incompressibility} (see also the full version in the master's thesis of Cai~\cite{cai2012master}), who took up an ambitious project of obtaining a complete classification of graphs $H$ on which edge modification problems for the property of being $H$-free admit polynomial kernels. The project was very successful: for instance, the situation for $3$-connected graphs $H$ is completely understood, and among trees there is only a finite number of remaining unresolved cases. In particular, the study of Cai and Cai revealed that the existence of a polynomial kernel for edge modification problems is actually a rare phenomenon that appears only for very simple graphs $H$.

One of the most tantalizing questions that is still unresolved is the case $H=K_{1,3}$, i.e., the {\sc{Claw-free Edge Deletion}} problem (as well as the completion and editing variants). The study of this particular case is especially interesting in light of the recent powerful decomposition theorem for claw-free graphs, proved by Chudnovsky and Seymour~\cite{ChudnovskyS08c}. For many related problems, having an equivalent structural view on the considered graph class played a crucial role in the design of a polynomial kernel, and hence there is hope for a positive result in this case as well. For this reason, determining the existence of a polynomial kernel for {\sc{Claw-free Edge Deletion}} was posed as an open problem during Workshop on Kernels (WorKer) in 2013, along with the same question for the related {\sc{Line Graph Edge Deletion}} problem~\cite{worker-opl}.

\medskip

\noindent{\bf{Our results.}} As an intermediate step towards showing a polynomial kernel for \textsc{Claw-free Edge Deletion}, we study a related variant, where we forbid \emph{diamonds} as well.\footnote{A more detailed discussion of the relation between these two problems is provided in the conclusions section.}
  By a {\em{diamond}} we mean a $K_4$ with one edge removed, and \clawdiamond-free graphs are exactly graphs that do not contain claws or diamonds as induced subgraphs. This graph class is equal to the class of line graphs of triangle-free graphs, and to the class of \emph{linear dominoes} (graphs in which every vertex is in at most two maximal cliques and every edge is in exactly one maximal clique)~\cite{KloksKM94,MetelskyT03}.

In this paper, we consider the \cdedgedeletionfull problem (\cdedgedeletion for short) where, given a graph $G$ and an integer $k$, one is asked to determine whether there exists a subset $F$ of the edges of $G$ with $|F|\leq k$ such that $G-F$ is \clawdiamond-free; such a set $F$ is also called an {\em{HDS}}. 

Our main result is that \cdedgedeletion admits a polynomial kernel.

\begin{theorem}\label{thm:realmain}
\cdedgedeletion admits a polynomial kernel.
\end{theorem}
In order to prove Theorem~\ref{thm:realmain}, we give a \emph{polynomial-time compression} of \cdedgedeletion into a problem in NP. 
By a polynomial-time compression into an unparameterized problem $R$ we mean a polynomial-time algorithm that, given an instance $(G,k)$ of \cdedgedeletion, outputs an equivalent instance $y$ of $R$ such that $|y|\leq f(k)$, for some computable function $f$ called the {\em{size}} of the compression. 

\begin{theorem}\label{thm:main}
\cdedgedeletion admits a polynomial-time compression algorithm into a problem in NP, where the size of the compression is $\Oh(k^{24})$.
\end{theorem}
The problem in NP that Theorem~\ref{thm:main} refers to actually is an annotated variant of \cdedgedeletion. Unfortunately, we are unable to express the annotations in a clean manner using gadgets. Therefore, we compose the polynomial-time compression of Theorem~\ref{thm:main} with the NP-hardness reduction that we present for \cdedgedeletion (see Corollary~\ref{cor:lower-bound} discussed below) in order to derive Theorem~\ref{thm:realmain}. 

To prove Theorem~\ref{thm:main}, we apply the vertex modulator technique. We first greedily pack edge-disjoint claws and diamonds in the input graph. If more than $k$ such obstacles can be packed, then we immediately infer that we are dealing with a no-instance. Otherwise, we obtain a set $X\subseteq V(G)$ with $|X|\leq 4k$ such that every induced claw and diamond in $G$ has at least one edge with both endpoints in $X$; in particular, $G-X$ is \clawdiamond-free. This means that we can start to examine the structure of $G-X$ understood as a line graph of a triangle-free graph: it consists of a number of maximal cliques (called henceforth {\em{bags}}) that can pairwise share at most a single vertex, and for two intersecting bags $B_1,B_2$ there is no edge between $B_1\setminus B_2$ and $B_2\setminus B_1$. Next, we prove that the neighborhood of every vertex $x\in X$ in $G-X$ is contained only in at most $2$ bags, which gives us at most $8k$ bags that are important from the viewpoint of neighborhoods of vertices in $X$. The crux of the proof lies in observing that an optimum deletion set $F$ consists only of edges that are close to these important bags. Intuitively, all the edges of $F$ lie either in important bags or in bags adjacent to the important ones. A more precise combinatorial analysis leads to a set $S\subseteq V(G)$ of size polynomial in~$k$ such that every edge of $F$ has both endpoints in $S$. After finding such a set $S$, a polynomial-time compression for the problem can be constructed using a generic argument that works for every edge modification problem with a finite list of forbidden induced subgraphs.

On a high level, our approach uses a vertex modulator technique that is similar to one used by Drange and Pilipczuk~\cite{DrangeP14} for their recent polynomial kernel for {\sc{Trivially Perfect Editing}}. However, since we are dealing with a graph class with fundamentally different structural properties, the whole combinatorial analysis of the instance with the modulator $X$ (which forms the main part of the paper) is also fundamentally different. We also remark that Cai~\cite{cai2012master} obtained a kernel for the {\sc{Diamond-free Edge Deletion}} problem with $\Oh(k^4)$ vertices. However, the techniques used in that result seem unusable in our setting: their core observation is that a diamond can either be already present in the original graph $G$ or be created by removing an edge of a $K_4$, and thus one can analyze an auxiliary `propagation graph' with diamonds and $K_4$s of the original graph $G$ as nodes. In our setting, we also forbid claws, and the core combinatorial properties of this propagation graph become much too complicated to handle.

Finally, we complement our positive result by proving that \cdedgedeletion is NP-hard and does not admit a subexponential-time parameterized algorithm unless the Exponential Time Hypothesis of Impagliazzo et al.~\cite{ImpagliazzoPZ01} fails.

\begin{theorem}\label{thm:lower-bound}
There exists a polynomial-time reduction that, given an instance~$\phi$ of \textsc{3Sat} with $n$ variables and $m$ clauses, outputs an instance $(G,k)$ of \cdedgedeletion such that (a) $(G,k)$ is a yes-instance if and only if $\phi$ is satisfiable, (b) $|V(G)|,k=\Oh(n+m)$, and (c) $\Delta(G)=6$.
\end{theorem}

\begin{corollary}\label{cor:lower-bound}
Even on instances with maximum degree $6$, \cdedgedeletion is NP-complete and does not admit algorithms with running time $2^{o(k)}\cdot |V(G)|^{\Oh(1)}$ or $2^{o(|V(G)|)}$ unless the Exponential Time Hypothesis fails.
\end{corollary}

Corollary~\ref{cor:lower-bound} shows that, contrary to recent discoveries for a number of edge modification problems related to subclasses of chordal graphs~\cite{bliznets2014interval,bliznets2014proper,drange2014exploring,fomin2012subexponential,ghosh2013faster}, \cdedgedeletion does not enjoy the existence of subexponential-time parameterized algorithms. The reduction of Theorem~\ref{thm:lower-bound} resembles constructions for similar edge modification problems (see e.g.~\cite{drange2014exploring,DrangeP14,komusiewicz2012cluster}): every variable is replaced by a cyclic variable gadget that has to be completely broken by the solution in one of two possible ways, and variable gadgets are wired together with constant-size clause gadgets that verify the satisfaction of the clauses. 

                                                                                                                                                                                                                                                                                                                                                                                                                                       \onlylncs{In this extended abstract, we provide an almost complete proof of Theorem~\ref{thm:main}, with some simpler proofs
and illustrating figures expelled to Appendices~\ref{app:omitted} and~\ref{app:figures}, respectively.
The proof of Theorem~\ref{thm:lower-bound} is provided in Appendix~\ref{sec:hardness}.}

\section{Preliminaries}\label{sec:prelims}
\onlyfull{%
\paragraph*{Graphs}
We consider finite, undirected, simple graphs $G$ with vertex set $V(G)$ and edge set $E(G)$.
Edges $\{u,v\}\in E(G)$ will be written as $uv$ for short.
For a subset of vertices $S\subseteq V(G)$, the \emph{subgraph of $G$ induced by $S$}, denoted $G[S]$, is the graph with vertex set $S$ and edge set $\{uv\in E(G) \mid u,v\in S\}$.
We write $G-S$ for $G[V(G)\setminus S]$.
For a subset of edges $F\subseteq E(G)$, we write $G-F$ for the subgraph of $G$ obtained by deleting $F$, that is, $V(G-F)=V(G)$ and $E(G-F)=E(G)\setminus F$.
Two disjoint sets $X,Y \subseteq V(G)$ are \emph{fully adjacent} if for every $x \in X$ and $y \in Y$, the vertices $x$ and $y$
are adjacent. If one of these sets is a singleton, say $X = \{v\}$, then we say that $v$ and $Y$ are fully adjacent.

For a vertex $v\in V(G)$, the (open) neighborhood $N_G(v)$ of $v$ is the set $\{u \mid uv \in E(G)\}$.
The \emph{closed neighborhood} $N_G[v]$ of $v$ is defined as $N_G(v)\cup\{v\}$.
For a subset of vertices $S\subseteq V(G)$, we denote by $E_G(S)$ the set of edges of $G$ with both endpoints in $S$.
In this work $N_G$ and $E_G$ will always pertain to the graph named $G$, so we drop the subscript.}

\onlylncs{%
In this work we use standard graph notation, and standard notation from parameterized complexity;
a brief overview is provided in Appendix~\ref{app:omitted}.}

\onlyfull{\paragraph*{Cliques, claws and diamonds}}
A \emph{clique} of $G$ is a set of vertices that are pairwise adjacent in $G$; we often identify cliques with the complete subgraphs induced by them.
A \emph{maximal clique} is a clique that is not a proper subset of any other clique.
A \emph{claw} is a graph on four vertices $\{c,u,v,w\}$ with edge set $\{cu,cv,cw\}$, called \emph{legs} of the claw; we call $c$ the \emph{center} of the claw, and $u,v,w$ the \emph{leaves} of the claw.
When specifying the vertices of a claw we always give the center first.
A diamond is a graph on four vertices $\{u,v,w,x\}$ with edge set $\{uv,uw,vw,vx,wx\}$.

\onlyfull{%
\paragraph*{Parameterized complexity}
\def\cQ{\ensuremath{\mathcal{Q}}}
Parameterized complexity is a framework for refining the analysis of a problem's computational complexity by defining an additional ``parameter'' as part of a problem instance.
Formally, a parameterized problem is a subset $\cQ$ of $\Sigma^*\times \mathbb{N}$ for some finite alphabet $\Sigma$.
The problem is fixed parameter tractable if there is an algorithm which solves an instance $(x,k)$ of the problem in time $f(k)\cdot |x|^c$, where $f:\mathbb{N}\to\mathbb{N}$ is any computable function and $c$ is any integer.
If $f(k)=2^{o(k)}$, we say the algorithm is a subexponential parameterized algorithm.
A {\em{kernelization algorithm}} for $\cQ$ is an algorithm that takes an instance $(x,k)$ of $\cQ$ and in time polynomial in $|x|+k$ outputs an equivalent instance $(x',k')$ (i.e., $(x,k)$ is in $\cQ$ if and only if $(x',k')$ is) such that $|x'|\leq g(k)$ and $k'\leq g(k)$ for some computable function $g$.
If the \emph{size} of the kernel $g$ is polynomial, we say that $\cQ$ admits a polynomial kernel. We can relax this definition to the notion of a {\em{compression algorithm}}, where the output is required to be an equivalent instance $y$ of some unparameterized problem $\cQ'$, i.e., $(x,k)\in \cQ$ if and only if $y\in \cQ'$. The upper bound $g(k)$ on $|y|$ will be then called the {\em{size}} of the compression.
We refer the reader to the books of Downey and Fellows~\cite{DowneyF99} and of Flum and Grohe~\cite{flum2006parameterized} for a more rigorous introduction.}

\onlyfull{\paragraph*{Forbidden induced subgraphs}}
Consider any finite family of graphs $\cH$.
A graph $G$ is \emph{$\cH$-free} if for every $H\in\cH$, $G$ does not contain $H$ as an induced subgraph\onlyfull{ ($H$ is not isomorphic to $G[S]$ for any $S\subseteq V(G)$)}.
An \emph{HDS ($\cH$-free deletion set)} for $G$ is a subset of edges $F\subseteq E(G)$ such that $G-F$ is $\cH$-free. Whenever we talk about a {\em{minimal}} HDS, we mean inclusion-wise minimality.
\textsc{$\cH$-free Edge Deletion} is the parameterized problem asking, for a graph $G$ and a parameter $k$, whether $G$ has an HDS of size at most $k$.
In \textsc{Annotated $\cH$-free Edge Deletion} we are additionally given a set $S\subseteq V(G)$ and the question is whether $G$ has an HDS of size at most $k$ that is contained in $E(S)$.

\onlyfull{\bigskip}

Let $(G,k)$ be an instance of \textsc{$\cH$-free Edge Deletion}.
Recall that we can easily find a subset $X$ of the \emph{vertices} of $G$ of size polynomial in $k$ such that (in particular) $G-X$ is $\cH$-free.
We refer to such a set as a \emph{modulator} of $G$.
The construction here is basically the same as in Lemma~3.3 of~\cite{DrangeP14}, and a slightly stronger construction based on the Sunflower Lemma can be found in \cite{FominSV13}.

\begin{lemma}\label{lem:greedy-modulator}\onlylncs{\appsign{}\footnote{The proof of all statements marked with \appsign{} can be found in Appendix~\ref{app:omitted}.}}
Let $c = \max\{|V(H)| : H\in\cH\}$.
Then one can in polynomial time either find a subset $X\subseteq V(G)$ of size at most $c\cdot k$ such that every induced $H\in\cH$ in $G$ has an edge in $E(X)$, or conclude that $(G,k)$ is a no-instance.
\end{lemma}
\onlyfull{
\begin{proof}
Let $\{H_1,H_2,\dots,H_m\}$ be an inclusion-wise maximal set of edge-disjoint induced subgraphs from $\cH$ in $G$ (such a set can be greedily found in polynomial time).
Since any HDS must contain an edge from each of the subgraphs, it must have size at least $m$.
If $m>k$, then we can conclude that $(G,k)$ is a no-instance.
Otherwise, let $X=\bigcup_i V(H_i)$ be the union of their vertex sets;
clearly $|X|\leq \sum_{i=1}^m |V(H_i)| \leq c\cdot k$.
By the maximality of our choice, every induced $H\in\cH$ in $G$ has an edge in common with one of $H_1,\dots,H_m$ and thus an edge in $E(X)$.
Hence $X$ satisfies the claim.
\end{proof}
}

\noindent We finish this section by showing that it suffices to find a set $S$ of vertices of size polynomial in $k$ such that every minimal solution (every minimal HDS of size at most~$k$) is contained in $E(S)$.
Given such a set, we can compress the \textsc{$\cH$-free Edge Deletion} instance in polynomial time to an instance of the annotated version with $\Oh(|S|^{c-1})$ vertices, where $c = \max\{|V(H)| : H\in\cH\}$ (we assume $c>1$, as otherwise the problem is trivial).
Since the annotated version is in NP (as an unparameterized problem), this compression, together with an algorithm to obtain $S$, concludes the proof of Theorem~\ref{thm:main}.
Note that we do not require inclusion-wise minimal HDSs of size larger than $k$ to be contained in $E(S)$.

\begin{lemma}\label{lem:compression-to-annotated}\onlylncs{\appsign{}}
There is an algorithm that, given an instance $(G,k)$ of \textsc{$\cH$-free Edge Deletion} and a set $S \subseteq V(G)$ such that every inclusion-wise minimal HDS of size at most $k$ is contained in $E(S)$,
outputs in polynomial time a set $U$, where $S\subseteq U \subseteq V(G)$ and $|U|\leq \Oh(|S|^{c-1})$, such that $(G,k)$ is a yes-instance if and only if 
$(G[U],S,k)$ is a yes-instance of \textsc{Annotated $\cH$-free Edge Deletion}.
\end{lemma}
\onlyfull{
\begin{proof}
Let $(G,k)$ and $S$ be as in the statement,
we construct $U$ in the following way.
Add all vertices of $S$ to $U$.
For every set $M$ of at most $c-1$ vertices (possibly empty) in $S$ and every subset $F\subseteq E(M)$, add to $U$ all vertices of up to one subgraph $H\in \cH$ induced in $G-F$ such that $V(H)\cap S = M$ (if there is more than one, choose any).
Clearly $S\subseteq U \subseteq V(G)$, $|U|\leq |S|+c^2 \cdot 2^{\binom{c-1}{2}} \cdot |S|^{c-1}$, and $U$ can be constructed in polynomial time.
We claim that $(G,k)$ is a yes-instance if and only if 
$(G[U],S,k)$ is a yes-instance of \textsc{Annotated $\cH$-free Edge Deletion}.

Suppose first that $(G,k)$ is a yes-instance, that is, there is a set of edges $F\subseteq E(G)$ of size at most $k$ such that $G-F$ is $\cH$-free; without loss of generality suppose that $F$ is inclusion-wise minimal. By the claim's assumption, $F$ is contained in $E(S)$.
Since $G[U]$ is an induced subgraph of $G$, $G[U]-F$ is also $\cH$-free.
Thus $(G[U], S, k)$ is a yes-instance of \textsc{Annotated $\cH$-free Edge Deletion}.

Suppose then that $(G[U], S, k)$ is a yes-instance of \textsc{Annotated $\cH$-free Edge Deletion}, that is, there is a set of edges $F\subseteq E(S)$ of size at most $k$ such that $G[U]-F$ is $\cH$-free.
We claim $F$ is an HDS of $G$, too.
Suppose that, to the contrary, $G-F$ has an induced subgraph $H\in\cH$. If $V(H)\subseteq S$, then because $S\subseteq U$ we would have that $H$ is an induced subgraph of $G[U]-F$, a contradiction.
Let then $M=V(H)\cap S$, and note that $|M|\leq c-1$.
Since $F\subseteq E(S)$, observe that the non-edges of $H$ deleted by $F$ are all in $E(M)$ and hence $H$ is an induced subgraph in $G-(F\cap E(M))$ as well.
By the construction of $U$, there is an induced subgraph $H'\in\cH$ in $G-(F\cap E(M))$ such that $V(H')\cap S = M$ and all of the vertices of $H'$ were added to $U$.
Since $F\setminus E(M)$ deletes only edges in $E(S)\setminus E(M)$, $H'$ is an induced subgraph in $G-F$ as well.
But all the vertices of $H'$ are in $U$, so $H'$ is an induced subgraph in $G[U]-F$, a contradiction.
This shows $F$ must be an HDS of $G$ and hence $(G,k)$ is a yes-instance of \textsc{$\cH$-free Edge Deletion}.
\end{proof}
}

\section{Kernel}\label{sec:kernel}
In this section, we prove Theorem~\ref{thm:main}. As discussed below the statement of Theorem~\ref{thm:main}, this yields the proof of Theorem~\ref{thm:realmain} and thus the kernel. 
Throughout, let $(G,k)$ to be an instance of \cdedgedeletionfull.

\onlyfull{\medskip\noindent}
We first define a simple decomposition of \clawdiamond-free graphs, which follows from the fact that they are precisely the line graphs of triangle-free graphs, as shown by Metelsky and Tyshkevich~\cite{MetelskyT03}.
For a \clawdiamond-free graph~$G'$, let $\cB(G')$ be the family of vertex sets, called \emph{bags},
\iflncs{containing every maximal clique of $G'$ and a singleton $\{v\}$ for each simplicial vertex $v$ of $G'$
	(i.e., each vertex whose neighborhood is a clique).}{containing:
\begin{itemize}
	\item every maximal clique of $G'$, and
	\item a singleton $\{v\}$ for each simplicial vertex $v$ of $G'$\\
	(i.e., each vertex whose neighborhood is a clique).
\end{itemize}}

\onlylncs{From the definitions in Section~3 and Theorem 5.2 of~\cite{MetelskyT03}, it follows that \clawdiamond-free graphs are precisely the \emph{linear r-minoes} for $r=2$, that is, graphs $G'$ such that every vertex belongs to at most two maximal cliques and every edge belongs to exactly one maximal clique. As a direct consequence, we obtain the following.}

\begin{lemma}\label{lem:bag-decomposition}\onlylncs{\appsign{}}
Let $G'$ be a \clawdiamond-free graph.
Consider the family $\cB(G')$ of bags of $G'$.
Then:
\begin{enumerate}[label=(\alph*)]
	\item\label{p:ver} every non-isolated vertex of $G'$ is in exactly two bags;
	\item\label{p:edg} for every edge $uv\in E(G')$ there is exactly one bag containing both $u$ and $v$;
	\item\label{p:inters} every two bags have at most one vertex in common;
	\item\label{p:nonadj} if two bags $A,B$ have a common vertex $v$, then there is no edge between $A-v$ and $B-v$.
\end{enumerate}
Moreover, $|\cB(G')|\leq |V(G')|+|E(G')|$ and the family $\cB(G')$ can be computed in polynomial time.
\end{lemma}
\onlyfull{
\begin{proof}From the definitions of Section~3 and Theorem 5.2 of~\cite{MetelskyT03} it follows that \clawdiamond-free graphs are precisely the \emph{linear r-minoes} for $r=2$, that is, graphs $G'$ such that every vertex belongs to at most two maximal cliques and every edge belongs to exactly one maximal clique.
In particular every edge of $G'$ is contained in exactly one bag, which proves \ref{p:edg}.

Let $v$ be any non-isolated vertex of $G'$.
If the neighborhood of $v$ is a clique in $G'$, then $N[v]$ is the only maximal clique containing $v$ -- hence $v$ is in exactly two bags: the maximal clique and the singleton $\{v\}$, by definition.
If the neighborhood of $v$ is not a clique, then $v$ has neighbors $a,b$ that are not adjacent -- hence $v$ is contained in at least two bags: the maximal clique containing $va$ and the (different) maximal clique containing $vb$. As $G'$ is a linear $2$-mino, $v$ is not in any other maximal clique. Since $v$ is not simplicial, by the definition of $\cB(G')$ we conclude that also in this case $v$ is in exactly two bags. This concludes the proof of \ref{p:ver}.

Since all bags induce cliques in $G'$, two different bags cannot have more than one vertex in common, as this would imply that an edge joining them is contained in both of them. This proves \ref{p:inters}.

Finally, if two bags $A,B$ had a common vertex $v$ and there was an edge between $a\in A-v$ and $b\in B-v$, then since $A$ is a maximal clique not containing $b$, there would be a vertex $a'\in A$ non-adjacent to $b$.
But then the vertices $a,a',b,v$ would induce a diamond subgraph in $G'$, a contradiction. This proves \ref{p:nonadj}.

To see that $|\cB(G')|\leq |V(G')|+|E(G')|$, note that every bag of $\cB(G')$ is either a singleton bag or it contains an edge. The number of singleton bags is bounded by $|V(G')|$, while the number of bags containing an edge is bounded by $|E(G')|$ due to \ref{p:edg}. In order to compute $\cB(G')$, it suffices to construct first singleton bags for all simplicial and isolated vertices, and then for every edge of $G$ add the unique maximal clique containing it, constructed in a greedy manner.
\maybeqed\end{proof}
}

Now run the algorithm of Lemma~\ref{lem:greedy-modulator} on instance $(G,k)$. In case the algorithm concludes it is a no-instance, we return a trivial no-instance of \textsc{Annotated \clawdiamond-free Edge Deletion} as the output of the compression. Otherwise, let $X$ to be the obtained modulator;
that is, $X$ is a subset of $V(G)$ of size at most $4k$ such that every induced claw and diamond in $G$ has an edge in $E(X)$.
In particular, $G-X$ is a $\clawdiamond$-free graph, so using Lemma~\ref{lem:bag-decomposition} we compute in polynomial time the family of bags $\cB(G-X)$.
When referring to bags, we will refer to $\cB(G-X)$ only, and implicitly use Lemma~\ref{lem:bag-decomposition} to identify, for each non-isolated vertex $v$ in $G-X$, \emph{the two bags containing $v$}, and for each edge $e$ of $G-X$, \emph{the bag containing $e$}.

Knowing the structure of $G-X$, we proceed by describing the adjacencies between $X$ and $G-X$. The following definition will play a central role.
For $x\in X$, we call a bag $B$ of $G-X$ \emph{attached} to $x$ if:
\begin{itemize}
	\item $B$ is fully adjacent to $x$, and
	\item if $B=\{v\}$ for some vertex $v$ which is not isolated in $G-X$, 
		then the other bag containing $v$ is not fully adjacent to $x$.
\end{itemize}
We call a bag \emph{attached} if it is attached to some $x\in X$.
The next two propositions show that adjacencies between $X$ and $G-X$ are fully determined by the attachment relation, see Figure~\ref{fig:attached}.

\onlyfull{\iflncs{\input{figureAttached-wg}}{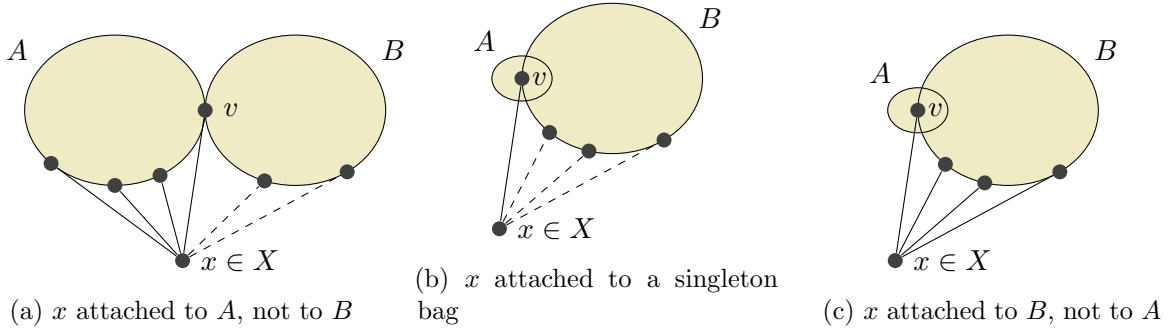
\begin{figure}[t]
\centering
\begin{subfigure}[b]{0.3\textwidth}
\centering
\begin{tikzpicture}
	\draw[bag] (-1.2,0) ellipse (1.2 and 1);
	\draw[bag] ( 1.2,0) ellipse (1.2 and 1);
	\node at (-2.5,0.8) {$A$};
	\node at (2.5,0.8) {$B$};
	\node[p,label=right:$x\in X$] (px) at (-0.3,-2) {};
	\node[p,label=right:$v$] (pv) at ($(1.2,0)+(180:1.2 and 1)$) {};
	\node[p] (p1) at ($(-1.2,0)+(-135:1.2 and 1)$) {};
	\node[p] (p2) at ($(-1.2,0)+(-90:1.2 and 1)$) {};
	\node[p] (p3) at ($(-1.2,0)+(-60:1.2 and 1)$) {};
	\node[p] (q1) at ($(1.2,0)+(-55:1.2 and 1)$) {};
	\node[p] (q2) at ($(1.2,0)+(-110:1.2 and 1)$) {};
	\draw (px) to (pv);
	\draw (px) to (p1) (px) to (p2) (px) to (p3);
	\draw[non] (px) to (q1) (px) to (q2);
\end{tikzpicture}
\caption{$x$ attached to $A$, not to $B$}
\end{subfigure}
\hspace{1em}
\begin{subfigure}[b]{0.29\textwidth}
\centering
\begin{tikzpicture}
	\draw[draw=none,bag] ( 1.2,0) ellipse (1.2 and 1);
	\draw[bag] (0,0) ellipse (0.4 and 0.3);
	\draw ( 1.2,0) ellipse (1.2 and 1);
	\node at (-0.5,0.5) {$A$};
	\node at (2.5,0.8) {$B$};
	\node[p,label=right:$x\in X$] (px) at (-0.3,-2) {};
	\node[p,label={[label distance=-3]right:$v$}] (pv) at ($(1.2,0)+(180:1.2 and 1)$) {};
	\node[p] (q1) at ($(1.2,0)+(-55:1.2 and 1)$) {};
	\node[p] (q2) at ($(1.2,0)+(-105:1.2 and 1)$) {};
	\node[p] (q3) at ($(1.2,0)+(-134:1.2 and 1)$) {};
	\draw (px) to (pv);
	\draw[non] (px) to (q1) (px) to (q2) (px) to (q3);
\end{tikzpicture}
\caption{$x$ attached to a singleton bag}
\end{subfigure}
\hspace{1em}
\begin{subfigure}[b]{0.27\textwidth}
\centering
\begin{tikzpicture}
	\draw[draw=none,bag] ( 1.2,0) ellipse (1.2 and 1);
	\draw[bag] (0,0) ellipse (0.4 and 0.3);
	\draw ( 1.2,0) ellipse (1.2 and 1);
	\node at (-0.5,0.5) {$A$};
	\node at (2.5,0.8) {$B$};
	\node[p,label=right:$x\in X$] (px) at (-0.3,-2) {};
	\node[p,label={[label distance=-3]right:$v$}] (pv) at ($(1.2,0)+(180:1.2 and 1)$) {};
	\node[p] (q1) at ($(1.2,0)+(-55:1.2 and 1)$) {};
	\node[p] (q2) at ($(1.2,0)+(-105:1.2 and 1)$) {};
	\node[p] (q3) at ($(1.2,0)+(-134:1.2 and 1)$) {};
	\draw (px) to (pv);
	\draw (px) to (q1) (px) to (q2) (px) to (q3);
\end{tikzpicture}
\caption{$x$ attached to $B$, not to $A$}
\end{subfigure}
\caption{Possible ways in which a vertex in $X$ can neighbor a vertex $v$ in $G-X$ and the two bags containing it.}
\label{fig:attached}
\end{figure}
}}

\begin{lemma}\label{lem:two-in-bag-implies-all}
Let $B\in\cB(G-X)$ be a bag such that some vertex $x\in X$ has at least two neighbors in $B$.
Then $B$ is attached to $x$. 
\end{lemma}
\begin{proof}
Suppose $x$ is adjacent to $u,v\in B$. 
If $x$ was non-adjacent to some vertex $w\in B$, 
then since $B$ induces a clique, the vertices $x,u,v,w$ would induce a diamond subgraph in $G$ (Figure~\ref{fig:one_bag_attached_proof} (a)).
However, no edge of this induced diamond would be in $E(X)$, contradicting the properties of $X$ as a modulator.
Therefore, all vertices of $B$ are adjacent to $x$ (and $|B|>1$), so $B$ is attached to $x$.
\maybeqed\end{proof}

\begin{lemma}\label{lem:one-bag-attached}
Let $v$ be a vertex in $G-X$ adjacent to a vertex $x\in X$.
Then there is exactly one bag in $\cB(G-X)$ that contains $v$ and is attached to $x$.
\end{lemma}
\begin{proof}
If $v$ is an isolated vertex in $G-X$, then $\{v\}$ is the only bag containing $v$ and is by definition attached to $x$.

Otherwise, let $A,B$ be the two bags containing $v$.
If one of the bags is a singleton, say $A=\{v\}$, then $B$, being unequal to $A$, contains some other vertices.
If at least one vertex of $B\setminus\{v\}$ is adjacent to $x$, then it follows from Lemma~\ref{lem:two-in-bag-implies-all} that $B$ is attached to $x$ and $A$ is not. Otherwise, i.e.~if no vertices of $B \setminus\{v\}$ are adjacent to $x$, then by definition $A$ is attached to $x$ and $B$ is not.

It remains to consider the case when both $A-v$ and $B-v$ are not empty; see Figure~\ref{fig:one_bag_attached_proof}, (b) and (c). Suppose that $x$ is adjacent to a vertex $a\in A-v$ and a vertex $b\in B-v$.
Then $a,b$ are non-adjacent by Lemma~\ref{lem:bag-decomposition}\ref{p:nonadj},
so vertices $v,a,b,x$ induce a diamond subgraph in $G$.
However, no edge of this diamond is in $E(X)$, a contradiction.

Suppose $x$ is non-adjacent to a vertex $a\in A-v$ and a vertex $b\in B-v$.
Then $a,b$ are non-adjacent by Lemma~\ref{lem:bag-decomposition}\ref{p:nonadj},
so vertices $x,a,b,v$ induce a claw subgraph in $G$.
However, no edge of this claw is in $E(X)$, again a contradiction.

Therefore, if $x$ is adjacent to a vertex in $A-v$, then 
$A$ is attached to $x$ (by Lemma~\ref{lem:two-in-bag-implies-all}) and $x$ must be non-adjacent to all of $B-v$, implying $B$ is not attached to $x$.
Otherwise, if $x$ is non-adjacent to all vertices in $A-v$, then $x$ must be adjacent to every vertex of $B-v$. This means $B$ is attached to $x$ and $A$ is not.
\maybeqed\end{proof}

\onlyfull{\iflncs{\input{figureOneBagAttachedProof-wg}}{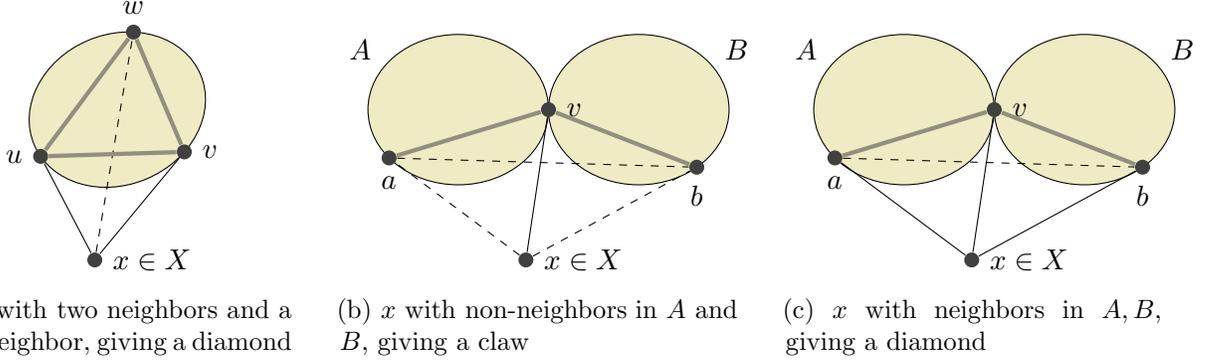
\begin{figure}[t]
\centering
\begin{subfigure}[b]{0.29\textwidth}
\centering
\begin{tikzpicture}
	\begin{scope}[rotate=23]
		\draw[bag] (0,0) ellipse (1.2 and 1);
		\node[p,label=left:$u$] (pu) at ($(-170:1.2 and 1)$) {};
		\node[p,label=right:$v$] (pv) at ($(-60:1.2 and 1)$) {};
		\node[p,label=above:$w$] (pw) at ($(60:1.2 and 1)$) {};
	\end{scope}
	\node[p,label=right:$x\in X$] (px) at (-0.3,-2) {};
	\draw    (px)--(pu) (px)--(pv);
	\draw[non] (px)--(pw);
	\draw[Ebag] (pu)--(pv) (pv)--(pw) (pw)--(pu);
\end{tikzpicture}
\caption{$x$ with two neighbors and a non-neighbor, giving a diamond}
\end{subfigure}
\hspace{1em}
\begin{subfigure}[b]{0.32\textwidth}
\centering
\begin{tikzpicture}
	\draw[bag] (-1.2,0) ellipse (1.2 and 1);
	\draw[bag] ( 1.2,0) ellipse (1.2 and 1);
	\node at (-2.5,0.8) {$A$};
	\node at (2.5,0.8) {$B$};
	\node[p,label=right:$x\in X$] (px) at (-0.3,-2) {};
	\node[p,label=right:$v$] (pv) at ($(1.2,0)+(180:1.2 and 1)$) {};
	\node[p,label=below:$a$] (pa) at ($(-1.2,0)+(-140:1.2 and 1)$) {};
	\node[p,label=below:$b$] (pb) at ($(1.2,0)+(-50:1.2 and 1)$) {};
	\draw (px) to (pv);
	\draw[non] (px) to (pb) (px) to (pa);
	\draw[non] (pa) to (pb);
	\draw[Ebag] (pv)--(pa) (pv)--(pb);
\end{tikzpicture}
\caption{$x$ with non-neighbors in $A$ and $B$, giving a claw}
\end{subfigure}
\hspace{1em}
\begin{subfigure}[b]{0.3\textwidth}
\centering
\begin{tikzpicture}
	\draw[bag] (-1.2,0) ellipse (1.2 and 1);
	\draw[bag] ( 1.2,0) ellipse (1.2 and 1);
	\node at (-2.5,0.8) {$A$};
	\node at (2.5,0.8) {$B$};
	\node[p,label=right:$x\in X$] (px) at (-0.3,-2) {};
	\node[p,label=right:$v$] (pv) at ($(1.2,0)+(180:1.2 and 1)$) {};
	\node[p,label=below:$a$] (pa) at ($(-1.2,0)+(-140:1.2 and 1)$) {};
	\node[p,label=below:$b$] (pb) at ($(1.2,0)+(-50:1.2 and 1)$) {};
	\draw (px) to (pv);
	\draw (px) to (pb) (px) to (pa);
	\draw[non] (pa) to (pb);
	\draw[Ebag] (pv)--(pa) (pv)--(pb);
\end{tikzpicture}
\caption{$x$ with neighbors in $A,B$, giving a diamond}
\end{subfigure}
\caption{Adjacencies between $X$ and $G-X$ that lead to a contradiction.}
\label{fig:one_bag_attached_proof}
\end{figure}
}}

We can now limit the number of attached bags by $2|X|$, which is linear in $k$.

\begin{lemma}\label{lem:leq2-bags-attached}
For any $x\in X$, there are at most two bags in $\cB(G-X)$ attached to $x$.
\end{lemma}
\begin{proof}
Let $x\in X$.
We first show that bags attached to $x$ must be pairwise disjoint and non-adjacent.
If two bags attached to $x$ contained a common vertex $v$, then $v$ would be adjacent to $x$ and, by Lemma~\ref{lem:one-bag-attached}, at most one of the bags would be attached to $x$, a contradiction.

If there was an edge $uv$ between two different bags attached to $x$,
then its endpoints $u$ and $v$ would be adjacent to $x$ and, by Lemma~\ref{lem:two-in-bag-implies-all}, the bag containing the edge $uv$ would be attached to $x$.
But we have just shown that bags attached to $x$ are disjoint, so no other bag attached to $x$ could contain $u$ or $v$, a contradiction.
Therefore, every two bags attached to $x$ are disjoint and non-adjacent.

Suppose that there are three or more bags adjacent to $x$.
Let $u,v,w$ be any vertices contained in three different bags.
By the above observations, $u,v,w$ are pairwise different and non-adjacent.
Hence, vertices $x,u,v,w$ induce a claw in $G$ that has no edges in $E(X)$,
a contradiction.
\maybeqed\end{proof}

Having limited the number of attached bags, we want to show that unattached bags intersect solutions only in a simple way.
The following technical proposition will help handle cases involving diamonds.

\begin{lemma}\label{lem:diamond_inside_bag}
Let $H$ be a subgraph (not necessarily induced) of $G$ isomorphic to a diamond.
Let $B\in\cB(G-X)$ be an unattached bag containing at least two vertices of $H$.
Then $B$ contains all vertices of $H$.
\end{lemma}
\begin{proof}
Let $u,v$ be two vertices of $H$ in $B$.
Let $w$ be a vertex of $H$ adjacent to $u$ and $v$ in $H$
(note that since $H$ is a diamond, there always is such a vertex).
Then $w$ is also adjacent to $u$ and $v$ in $G$.
Vertex $w$ cannot be in $X$, as otherwise Lemma~\ref{lem:two-in-bag-implies-all} would contradict the assumption that $B$ is unattached.
Hence, $w$ is in $G-X$. 
Let $A$ be the bag containing the edge $uw$.
If $w$ was not in $B$, then $B\neq A$ and $vw$ would be an edge going between $v\in B-u$ and $w\in A-u$,
contradicting Lemma~\ref{lem:bag-decomposition}\ref{p:nonadj}.
Therefore, $w\in B$.

Repeating this argument for the fourth vertex of the diamond $H$ and an appropriate pair of vertices from $\{u,v,w\}$, all the vertices of $H$ can be shown to be in $B$.
\maybeqed\end{proof}

It turns out that one may need to delete an edge of an unattached bag $B$, but in this case the intersection of any minimal HDS $F$ with the edges of $B$ has a very special structure: deleting the edges of $F$ makes some of the vertices of $B$ isolated, whereas the rest of $B$ remains a smaller clique.
This will later allow us to take only a limited number of unattached bags into account.

\begin{lemma}\label{lem:all-or-nothing}
Let $F$ be a minimal HDS of $G$ and let $B\in\cB(G-X)$ be an unattached bag.
Then $G[B]-F$ consists of a clique and a number of isolated vertices.
\end{lemma}
\begin{proof}
Let $B'\subseteq B$ be the set of vertices that are not isolated in $G[B]-F$.
Consider the set $F'=F\setminus E(B')$. The graph $G-F'$ is obtained from $G-F$ by adding back all edges between vertices in $B'$.
Thus the bag $B$ induces in $G-F'$ a clique on $B'$ plus isolated vertices $B\setminus B'$.
We claim that $F'$ is an HDS.
By the minimality of $F$, this will imply that $F=F'$ and hence the claim.

Suppose to the contrary that $G-F'$ contains an induced claw or diamond $H$.
Since $G-F$ contains neither an induced claw nor a diamond, $H$ has an edge $e$ in $F\cap E(B')$.

If $H$ is a diamond in $G-F'$, then since $e$ has both endpoints in $B$, by Lemma~\ref{lem:diamond_inside_bag} we infer that all vertices of $H$ are in $B$.
But this contradicts that $B$ induces a clique plus isolated vertices in $G-F'$.

If $H$ is a claw in $G-F'$, then let $c$ be its center and $v,u_1,u_2$ its leaves, so that $e=cv$.
Since $e\in E(B')$, its endpoint $c$ is in $B'$, meaning $c$ is not isolated in $G[B]-F$.
Let $w$ be a neighbor of $c$ in $G[B]-F$.
We show that vertices $c,w,u_1,u_2$ induce a claw in $G-F$.
Consider where the leaves $u_i$ may be.
If $u_i\in B$ (for $i=1$ or $2$), then vertices $c,v,u_i$ induce two legs of a claw (a $P_3$) in $G[B]-F'$, contradicting that $G[B]-F'$ is a clique plus isolated vertices.
If $u_i \in X$, then since $u_i$ is adjacent to $c\in B$ and $B$ is not attached, by Lemma~\ref{lem:two-in-bag-implies-all} we have that $u_i$ cannot be adjacent to $w\in B$ in $G$.
If $u_i\in G-(B \cup X)$, then it is in the bag $A$ containing the edge $cu_i$
and, by Lemma~\ref{lem:bag-decomposition}, $u_i\in A-c$ is not adjacent to $w\in B-c$ in $G$.
In either case $u_1 w$ and $u_2 w$ are non-edges in $G$, thus also in $G-F$.
By assumption, $u_1u_2$ is a non-edge in $G-F'$, thus also in $G-F$.
We showed that $u_i\not\in B$, so $cu_i\in E(G-F')$ are also edges in $G-F$.
Finally, $cw\in E(G[B]) \setminus F$, so indeed the vertices $c,w,u_1,u_2$ induce a claw in $G-F$, a contradiction.
\maybeqed\end{proof}

\onlyfull{To obtain a compressed description of the problem, one ingredient remains: limiting the size of bags that may need deletions.}

\begin{lemma}\label{lem:no-big-bag-deletions}
If $K$ is a clique in $G$ with at least $2k+2$ vertices,
then every HDS $F$ of $G$ of size at most $k$ satisfies $F \cap E(K) = \emptyset$.
\end{lemma}
\begin{proof}
By contradiction, assume there exists $uv \in F$ with $u,v \in K$.
However, then for every two distinct $w_1,w_2 \in K \setminus \{u,v\}$, the subgraph induced in $G-uv$ by $u,v,w_1,w_2$ is a diamond.
As $|K| \geq 2k+2$, we can find $k$ edge-disjoint diamonds formed in this way in $G-uv$.
Consequently, $F$ needs to contains at least $k$ edges apart from $uv$, a contradiction.
\maybeqed\end{proof}

\begin{corollary}\label{cor:no-big-bag-deletions}\onlylncs{\appsign{}}
Let $B\in\cB(G-X)$ be a bag with at least $2k+2$ elements.
Then for every HDS $F$ of $G$ of size at most $k$, 
$F\cap E(B) = \emptyset$.
If furthermore $B$ is attached to $x\in X$, then 
$F\ \cap\ E(B\cup\{x\}) = \emptyset$.
\end{corollary}
\onlyfull{\begin{proof}
Follows directly from Lemma~\ref{lem:no-big-bag-deletions}, since every bag $B$ is a clique,
if $B$ is attached to $x \in X$, then $B \cup \{x\}$ is a clique as well.
\maybeqed\end{proof}}
We are ready to present the main step of the compression procedure for \cdedgedeletion.

\begin{lemma}\label{lem:compression}
One can in polynomial time find a set $S\subseteq V(G)$ of size $\Oh(k^4)$ such that every minimal HDS of size at most $k$ is contained in $E(S)$.
\end{lemma}
\begin{proof}
Call a bag \emph{small} if it has less than $2k+2$ vertices, \emph{big} otherwise.
\iflncs{%
We \emph{mark} every small attached bag,
every small unattached bag that shares a vertex with some small attached bag,
and furthermore, for every vertex pair $x,y\in X$, we mark up to $k+1$ small unattached bags of size at least two that have a vertex in $N(x)\cap N(y)$.}{%
We \emph{mark} the following bags:
\begin{itemize}
\item every small attached bag,
\item every small unattached bag that shares a vertex with some small attached bag,
\item for every vertex pair $x,y\in X$, we mark up to $k+1$ small unattached bags of size at least two that have a vertex in $N(x)\cap N(y)$ (if there are more such bags, we mark any $k+1$ of them).
\end{itemize}}
Let $S$ be the set of all vertices in marked bags and in $X$.
\iflncs{%
A direct calculation, provided in Appendix~\ref{app:omitted}, shows that $|S| = \Oh(k^4)$.}{%
Let us first show that $|S|=\Oh(k^4)$.
By the construction of $X$ in Lemma~\ref{lem:greedy-modulator}, we have that $|X|\leq 4k$.
By Lemma~\ref{lem:leq2-bags-attached}, there are at most $2|X|$ attached bags.
Hence, there are at most $2|X| \cdot (2k+1)$ vertices in small attached bags.
Since each vertex of $G-X$ is in at most two bags, there are at most $2|X| \cdot (2k+1)$ small unattached bags that share a vertex with small attached bags.
In the final point we mark at most $|X|^2\cdot (k+1)$ small bags.
Therefore, we mark at most $2|X|+2|X|\cdot (2k+1)+|X|^2\cdot(k+1)=\Oh(k^3)$ small bags in total.
The set $S\setminus X$ contains at most $(2k+1)$ times as many vertices in total, which together with $|X|\leq 4k$ implies that $|S|=\Oh(k^4)$.}

\onlyfull{\smallskip}

We want to show that a minimal HDS never deletes any edges in unmarked bags.
Let $Z$ be the set of edges that are either contained in a marked bag, or in $E(X)$, or connect a vertex of a marked bag with a vertex of $X$.
Note that $Z \subseteq E(S)$, but the inclusion may be strict, due to an edge going between two vertices of some marked bags that belongs to an unmarked bag.
Let $F$ be a minimal HDS of size at most $k$.
\iflncs{We show that $F' = F \cap Z$ is also an HDS, concluding the proof of the lemma.}{%
We will show that $F'=F\cap Z$ is also an HDS.
By the minimality of $F$, this will imply that $F=F'\subseteq Z \subseteq E(S)$, and hence the proof of the lemma will be concluded.}

\begin{myclaim}\label{cl:clique-indset-Fp}
If a bag does not induce a clique plus isolated vertices in $G-F'$,
then it is a small attached bag.
\end{myclaim}
\begin{proof}
First consider $G-F$. By Lemma~\ref{lem:all-or-nothing}, every unattached bag induces a clique plus isolated vertices in $G-F$.
By Corollary~\ref{cor:no-big-bag-deletions}, every big bag induces a clique in $G-F$.
Hence, if a bag does not induce a clique plus isolated vertices in $G-F$,
then it is a small attached bag.
Suppose now that a bag does not induce a clique plus isolated vertices in $G-F'$. Then it necessarily contains an edge of $F'\subseteq Z$ and thus must be marked.
We infer that this bag induces the same subgraph in $G-F$ as in $G-F'$. Therefore, it must be small and attached.
\cqed\maybeqed\end{proof}

Suppose to the contrary that $G-F'$ contains an induced claw or diamond $H$.
Since $G-F$ contained none, $H$ must have an edge $e \in F\setminus F' = F\setminus Z$. We consider the following cases depending on the location of $e$, each leading to a contradiction; see Figure~\ref{fig:compression_proof}\onlylncs{ in Appendix~\ref{app:figures}}.

\medskip
\noindent\textbf{Case 1:} edge $e$ has an endpoint in the modulator $X$.\\
Then $e=vx$ for some $x\in X$ and $v\in V(G)$.
If $v \in X$, then $e\in E(X)\subseteq Z$, contradicting $e\in F\setminus Z$.
Otherwise, by Lemma~\ref{lem:one-bag-attached}, 
there is a bag $B$ containing $v$ that is attached to $x$.
Since $e\in F$, by Corollary~\ref{cor:no-big-bag-deletions} we infer that $B$ has less than $2k+2$ elements.
But then $B$ is a small, attached, and hence marked bag, implying $e\in Z$, a contradiction.

\medskip
\noindent\textbf{Case 2:} edge $e$ has both endpoints in $G-X$ (and thus $e$ is in $G-X$).\\
Let $B$ be the bag containing $e$.
Since $e\in F$, $B$ is a small bag by Corollary~\ref{cor:no-big-bag-deletions}.
Since $e\not\in Z$, $B$ is not a marked bag.
Since small attached bags are marked, $B$ is unattached.
By Claim~\ref{cl:clique-indset-Fp}, $B$ induces a clique plus isolated vertices in $G-F'$.

\medskip
\noindent\textbf{Case 2a:} $H$ is a diamond (in $G-F'$).\\
Then the endpoints of $e$ are in $B$, hence by Lemma~\ref{lem:diamond_inside_bag} all vertices of $H$ are in $B$.
But $B$ induces a clique plus isolated vertices in $G-F'$, a contradiction.

\medskip
\noindent\textbf{Case 2b:} $H$ is a claw (in $G-F'$).\\
Let $c$ be the center of the claw $H$ and let $v,u_1,u_2$ be its leaves, so that $e=cv$.
Let $A$ be the other bag containing $c$.

If $u_i$ was in $B$ (for $i=1$ or $2$), then $B$ would not induce a clique plus isolated vertices in $G-F'$ because $u_i,c,v$ induces a $P_3$, a contradiction.

If $u_i\not\in X$, then $u_i$ is in the bag containing $cu_i$ but not in $B$, which means that $u_i$ is in $A$.
If both $u_1,u_2$ were not in $X$, then $A$ would not induce a clique plus isolated vertices in $G-F'$ (because $u_1,c,u_2$ induces a $P_3$). By Claim~\ref{cl:clique-indset-Fp}, $A$ would be a small attached bag that shares the vertex $c$ with $B$, implying that $B$ is marked, a contradiction.

If exactly one leaf of the claw is in $X$, e.g., $u_1 \in X$ and $u_2\in G-X$, then $u_2$ is in $A$ (as above).
Because $c$ is adjacent to $u_1\in X$, by Lemma~\ref{lem:one-bag-attached} we infer that one of $A,B$ is attached to $u_1$.
Since $B$ is unattached, $A$ is attached to $u_1$, so $u_1u_2$ is an edge in $G$. Since $u_1u_2$ is not an edge in $G-F'$, we have that $u_1u_2\in F'\subseteq F$.
By Corollary~\ref{cor:no-big-bag-deletions} we infer that $A$ is a small bag.
It is also attached, and therefore $B$ is marked, again a contradiction.

If both $u_1,u_2$ are in $X$, then note that $B$ is an unattached bag of size at least two that has a vertex (namely $c$) in the common neighborhood of $u_1$ and $u_2$.
By the definition of marked bags and as $B$ was not marked in the third point, at least $k+1$ different marked bags $B_1,\dots,B_{k+1}$ are unattached, have size at least two, and have some vertex, respectively $c_1,c_2,\dots,c_{k+1}$, in the common neighborhood of $u_1$ and $u_2$.
If $c_i=c_j$ for some $i,j$ with $1\leq i<j\leq k+1$, then $B_i,B_j$ are the two bags that contain $c_i$. Since $c_i$ is adjacent to $u_1$, one of those bags is attached to $u_1$ by Lemma~\ref{lem:one-bag-attached}, a contradiction. Hence, $c_i\neq c_j$ for all $1\leq i<j\leq k+1$.
Let $w_i$ be any vertex different from $c_i$ in $B_i$.
Since $B_i$ is unattached, $w_i$ is non-adjacent to $u_1$ and $u_2$ in $G$ by Lemma~\ref{lem:two-in-bag-implies-all}.
Clearly, $c_i$ is adjacent to $w_i,u_1,u_2$ in $G$.
Therefore, vertices $c_i,w_i,u_1,u_2$ induce $k+1$ edge-disjoint claws in $G-u_1u_2$.
Since $u_1,u_2$ are leaves of the claw $H$ in $G-F'$, they are non-adjacent in $G-F$.
Hence, for each $i$ with $1\leq i\leq k+1$, one of the edges $c_i w_i,c_i u_1,c_i u_2$ must be deleted by $F$.
But $|F|\leq k$, a contradiction.
\maybeqed\end{proof}
\onlyfull{\iflncs{\input{figureCompressionProof-wg}}{\begin{figure}[t]
\centering

\begin{subfigure}[b]{0.33\textwidth}
\centering
\begin{tikzpicture}
\draw[bagl] (-1.2,0) ellipse (1.2 and 1);
\draw[bagl] ( 1.2,0) ellipse (1.2 and 1);
\node at (2.8,0.8) {$<2k+2$};
\node[p,label=right:$x\in X$] (pc) at (-0.3,-2) {};
\node[p,label=right:$v$] (p1) at ($(1.2,0)+(180:1.2 and 1)$) {};
\node[s] (p2) at ($(1.2,0)+(-90:1.2 and 1)$) {};
\node[s] (p3) at ($(1.2,0)+(-50:1.2 and 1)$) {};
\node[s] (p4) at ($(1.2,0)+(-110:1.2 and 1)$) {};
\node[s] (p5) at ($(1.2,0)+(-130:1.2 and 1)$) {};
\draw[E] (pc) to (p1);
\draw    (pc)--(p2) (pc)--(p3) (pc)--(p4) (pc)--(p5);
\end{tikzpicture}
\caption{$e=vx$ has an endpoint in $X$}
\end{subfigure}
\begin{subfigure}[b]{0.3\textwidth}
\centering
\begin{tikzpicture}[scale=0.9]
\draw[bagl] (0,0) ellipse (1.2 and 1);
\node[p] (p1) at ($(0,0)+(  10:1.2 and 1)$) {};
\node[p] (p2) at ($(0,0)+( 100:1.2 and 1)$) {};
\node[p] (p3) at ($(0,0)+(190:1.2 and 1)$) {};
\node[p] (p4) at ($(0,0)+(-80:1.2 and 1)$) {};
\draw[E] (p1)--(p2) (p2)--(p3) (p3)--(p4) (p4)--(p1) (p2)--(p4);
\draw[E,non] (p1)--(p3);
\node at (0,-1.5) {};
\end{tikzpicture}
\caption{diamond with edge $e$ in a bag}
\end{subfigure}
\begin{subfigure}[b]{0.3\textwidth}
\centering
\begin{tikzpicture}
\draw[bagl] (0,0) ellipse (1.2 and 1);
\node[p,label=right:$c$] (p1) at ($(0,0)+(  10:1.2 and 1)$) {};
\node[p,label=left:$v$] (p2) at ($(0,0)+( 130:1.2 and 1)$) {};
\node[p,label=left:$u_i$] (p3) at ($(0,0)+(-110:1.2 and 1)$) {};
\draw[E] (p1)--(p2) (p1)--(p3);
\draw[E,non] (p2)--(p3);
\node at (0,-1.5) {};
\end{tikzpicture}
\caption{claw with two legs in $B$}
\end{subfigure}
\\
\begin{subfigure}[t]{0.26\textwidth}
\centering
\begin{tikzpicture}[scale=0.85,rotate=-20]
\draw[bagl] (-1.2,0) ellipse (1.2 and 1);
\draw[bagl] (1.2,0) ellipse (1.2 and 1);
\node[p,label={[label distance=3]-180:$c$}] (p1) at (0,0) {};
\node[p,label=below:$v$] (p4) at ($(-1.2,0)+( 130:1.2 and 1)$) {};
\node[p,label=right:$u_1$] (p2) at ($(1.2,0)+(50:1.2 and 1)$) {};
\node[p,label=right:$u_2$] (p3) at ($(1.2,0)+(-20:1.2 and 1)$) {};
\draw[E] (p1)--(p4) (p1)--(p2) (p1)--(p3);
\draw[E,non] (p2)--(p3);
\node at (0,-3) {};
\end{tikzpicture}
\caption{claw with two legs in $A$}
\end{subfigure}
\hspace{1em}
\begin{subfigure}[t]{0.26\textwidth}
\centering
\begin{tikzpicture}[scale=0.85]
\draw[bagl] (-1.2,0) ellipse (1.2 and 1);
\draw[bagl] (1.2,0) ellipse (1.2 and 1);
\node at (2.9,0.8) {$<2k+2$};
\node[p,label={[label distance=4]left:$c$}] (c) at (0,0) {};
\node[p,label=below:$v$] (v) at ($(-1.2,0)+( 130:1.2 and 1)$) {};
\node[p,label=right:$u_2$] (u2) at ($(1.2,0)+(10:1.2 and 1)$) {};
\node[p,label=right:$u_1\in X$] (u1) at (0.3,-2) {};
\node[s] (p5) at ($(1.2,0)+(-45:1.2 and 1)$) {};
\node[s] (p6) at ($(1.2,0)+(-105:1.2 and 1)$) {};
\node[s] (p7) at ($(1.2,0)+(-130:1.2 and 1)$) {};
\draw[E] (c)--(v) (c)--(u2) (c)--(u1);
\draw (u1)--(p5) (u1)--(p6) (u1)--(p7);
\draw[E,non] (u1)--(u2);
\node at (0,-3) {};
\end{tikzpicture}
\caption{claw with one leg in $A$}
\end{subfigure}
\hspace{1em}
\begin{subfigure}[t]{0.4\textwidth}
\centering
\begin{tikzpicture}[scale=0.95]
\begin{scope}[rotate around={45:(-1.3,-1.6)}]
	\draw[bagl] (-1.3,-1.6) ellipse (0.8 and 0.5);	\node[p,label={left:$v$}] (pv) at ($(-1.3,-1.6)+(-162:0.8 and 0.5)$) {};
	\node[p,label=170:$c$] (pc) at ($(-1.3,-1.6)+(0:0.8 and 0.5)$) {};
\end{scope}
\begin{scope}[rotate around={-60:(1.7,0)}]
	\draw[bagl] (1.7,0) ellipse (1 and 0.8);
\end{scope}
\begin{scope}[rotate=60]
	\draw[bagl] (0,0) ellipse (1 and 0.8);
	\node[p,label={right:$w_1,w_2$}] (pw1) at ($(0,0)+(  -55:1 and 0.8)$) {};
	\node[p,label=above:$c_1$] (pc1) at ($(0,0)+(-150:1 and 0.8)$) {};
\end{scope}
\begin{scope}[rotate around={-60:(1.7,0)}]
	\node[p,label=above:$c_2$] (pc2) at ($(1.7,0)+(-10:1 and 0.8)$) {};
\end{scope}
\begin{scope}[rotate around={60:(3.5,-0.2)}]	
	\draw[bagl] (3.5,-0.2) ellipse (0.8 and 0.5);
	\node[p,label=above:$c_3$] (pc3) at ($(3.5,-0.2)+(180:0.8 and 0.5)$) {};
	\node[p,label=right:$w_3$] (pw3) at ($(3.5,-0.2)+(-15:0.8 and 0.5)$) {};
\end{scope}
\begin{scope}[rotate around={30:(3.5,-2)}] 
	\draw[bagl] (3.5,-2) ellipse (0.8 and 0.5);
	\node[p,label=60:$c_4$] (pc4) at ($(3.5,-2)+(180:0.8 and 0.5)$) {};
	\node[p,label=right:$w_4$] (pw4) at ($(3.5,-2)+(-15:0.8 and 0.5)$) {};
\end{scope}

\node[p,label={[label distance=8]left:$u_1\in X$}] (pu1) at (0.6,-3.2) {};
\node[p,label=right:$u_2\in X$] (pu2) at (1.8,-3.2) {};
\draw[E] (pc)--(pv) (pc)--(pu1) (pc)--(pu2);
\draw    (pc1)--(pw1) (pc1)--(pu1) (pc1)--(pu2)
         (pc2)--(pw1) (pc2)--(pu1) (pc2)--(pu2)
         (pc3)--(pw3) (pc3)--(pu1) (pc3)--(pu2)
         (pc4)--(pw4) (pc4)--(pu1) (pc4)--(pu2);
\draw[E,non] (pu1)--(pu2) (pu1)--(pv) (pu2)--(pv);
\end{tikzpicture}
\caption{claw with two leaves in $X$}
\end{subfigure}
\caption{The different situations where a claw or diamond (thick blue edges and dashed non-edges) might appear in $G-F'$, each leading to a contradiction.}
\label{fig:compression_proof}
\end{figure}
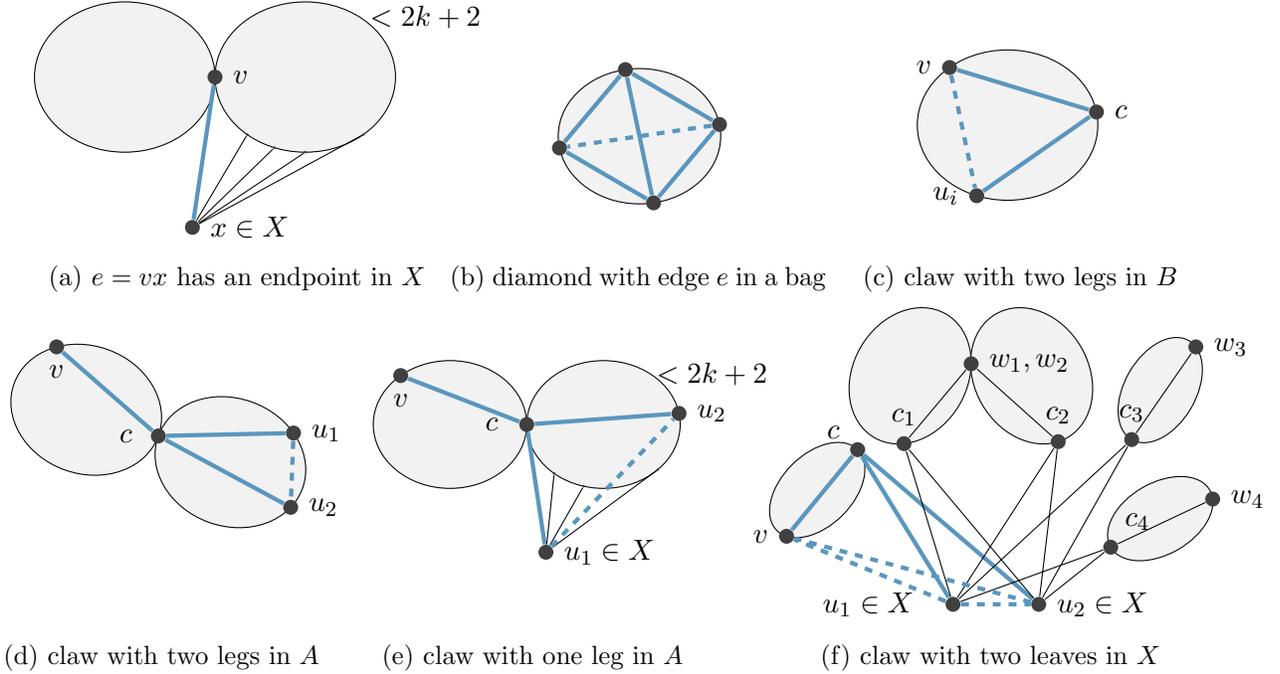 
}}

To prove Theorem~\ref{thm:main}, we now plug $(G,k)$ and the set $S$ obtained in Lemma~\ref{lem:compression} into the construction of Lemma~\ref{lem:compression-to-annotated} in order to obtain the polynomial-time compression. Observe that the resulting instance of \textsc{Annotated $\cH$-free Edge Deletion} has $\Oh(k^{12})$ vertices and $\Oh(k^{24})$ edges. 

\section{Hardness}\label{sec:hardness}
In this section, we prove Theorem~\ref{thm:lower-bound}, which states that the problems we consider cannot be solved in subexponential time, under the Exponential Time Hypothesis (ETH).
Let us recall that this hypothesis, formulated by Impagliazzo, Paturi and Zane~\cite{ImpagliazzoPZ01}, states that there exists a positive real number $s$  such that \textsc{3Sat} with $n$ variables cannot be solved in time $\Oh(2^{sn})$.
The Sparsification Lemma of~\cite{ImpagliazzoPZ01} allows to strengthen this assumption to functions subexponential in the size of the formula (the number variables $n$ plus the number of clauses $m$ of the input formula), and not just the number of variables. More precisely, unless ETH fails, \textsc{3Sat} cannot be solved in time $\Oh(2^{s(n+m)})$ for some $s>0$. In Theorem~\ref{thm:lower-bound}, we give a reduction that, given a \textsc{3Sat} instance $\phi$, outputs in polynomial time an equivalent instance $(G,k)$ of \cdedgedeletion where $k$ (the number of allowed deletions) is linear in the size of $\phi$. Composing this reduction with any subexponential parameterized algorithm for the problem would imply a subexponential algorithm for \textsc{3Sat}, contradicting ETH; this shows how Theorem~\ref{thm:lower-bound} implies Corollary~\ref{cor:lower-bound}

Our approach to proving Theorem~\ref{thm:lower-bound} is to consider \textsc{Claw-free Edge Deletion} in graphs were diamonds are not present and cannot appear after any edge deletions. That is, we shall actually prove the following result.

\begin{theorem}
\label{thm:claw_ETH}
There exists a polynomial-time reduction that, given an instance $\phi$ of \textsc{3Sat} with $n$ variables and $m$ clauses, outputs an instance $(G,k)$ of \textsc{Claw-free Edge Deletion} such that (a) $(G,k)$ is a yes-instance if and only if $\phi$ is satisfiable, (b) $|V(G)|,k=\Oh(n+m)$, (c) $G$ is $\{K_4,\mbox{diamond}\}$-free, and (d) $\Delta(G)=6$.
\end{theorem}

Theorem~\ref{thm:lower-bound} follows, since an instance of \textsc{Claw-free Edge Deletion} with no diamond or $K_{4}$ as an induced subgraph is a yes-instance if and only if it is a yes-instance of \cdedgedeletion (the solution sets are even identical, because deleting an edge from a $\{K_4,\mbox{diamond}\}$-free graph cannot create a diamond). Note that since $\Delta(G)=6$, both in Theorem~\ref{thm:claw_ETH} and in Theorem~\ref{thm:lower-bound}, we have that $|E(G)|\leq 3|V(G)|$ and under ETH there is even no subexponential-time algorithm in terms of the number of edges of the graph.
We remark that the original NP-hardness reduction for \textsc{Claw-free Edge Deletion} of Yannakakis~\cite{Yannakakis81}\footnote{Yannakakis~\cite{Yannakakis81} proves NP-hardness of \textsc{Line Graph Edge Deletion}, but the same reduction works also for \textsc{Claw-free Edge Deletion}.} actually implies that this problem cannot be solved in subexponential parameterized time; however, the constructed graph contains a lot of diamonds and the reduction cannot be easily adapted to our setting.

The remainder of this section is devoted to the proof of Theorem~\ref{thm:claw_ETH}.

Let $\phi$ be an instance of \textsc{3Sat} -- a formula in conjunctive normal form whose every clause has exactly three literals with three different variables (formulas with clauses of at most three, possibly equal literals can easily be transformed to this form via standard reductions; see e.g.~\cite{fomin2012subexponential}).
Let $\cV(\phi)$ be the set of variables of $\phi$; write $c\in \phi$ for clauses $c$ of $\phi$ and write $x\in \cV(c)$ for the three variables occurring in the clause. We write $\bot,\top$ for the \texttt{false},\texttt{true} values assigned to the variables, respectively.

Let us begin by defining the clause gadget for each clause $c\in\phi$.
We construct the following graph $G_c$ (see Figure~\ref{fig:gadgets}).
$G_c$ has 19 vertices: a central vertex $u_c$ and
vertices $v_c^x$, ${\tilde{v}}_c^x$, $w_c^x$, ${\tilde{w}}_c^x$, $t_c^x$, ${\tilde{t}}_c^x$ for each $x\in \cV(c)$;
and 27 edges arranged in 9 triangles:
$\{u_c,v_c^x,{\tilde{v}}_c^x\}$,
$\{v_c^x,w_c^x,{\tilde{w}}_c^x\}$,
$\{v_c^x,t_c^x,{\tilde{t}}_c^x\}$ for each $x\in \cV(c)$.

Not satisfying a literal of variable $x\in \cV(c)$ will correspond to deleting the edge $t_c^x{\tilde{t}}_c^x$ (thick blue edges on Figure~\ref{fig:gadgets}).
We show that in a minimum HDS at least one thick edge shall not be deleted.

\begin{myclaim}
\label{claim:clause_gadget}
Let $F$ be an HDS of $G_c$.
Then $F$ deletes at least 7 edges in $E(G_c)\setminus\{t_c^x {\tilde{t}_c}^x  \mid x\in \cV(c)\}$.
Furthermore, if it deletes exactly 7 edges in this set,
then $F$ does not delete $t_c^x {\tilde{t}}_c^x$ for some $x\in \cV(c)$.
\end{myclaim}
\begin{proof}
We drop the subscript $c$ of all vertices for clarity.
To see that $F$ deletes at least 7 edges in  $E(G_c)\setminus\{t^x {\tilde{t}}^x \mid x\in \cV(c)\}$ it suffices to notice that this set contains all edges of 7 edge-disjoint induced claw subgraphs of $G_c$:
$\{u,{\tilde{v}}^x \mid x\in \cV(c)\}$ and
$\{v^x,u,w^x,t^x\}$,
$\{v^x,{\tilde{v}}^x,{\tilde{w}}^x,{\tilde{t}}^x\}$ for $x\in \cV(c)$.

Suppose now that $F$ deletes exactly 7 edges in the above set and deletes all edges $t^x {\tilde{t}}^x$ for $x\in \cV(c)$.
Then it must delete exactly one edge of each of the 7 claws above.
In particular, for every $x \in \cV(c)$ it does not delete both $v^x u$ and $v^x w^x$, which belong to the claw $\{v^x,u,w^x,t^x\}$.
Then, it needs to delete $v^x t^x$ or $v^x {\tilde{t}}^x$, as these edges form a claw both with $v^xu$ and with $v^x w^x$.
Consider now claws $\{v^x,u,w^x,t^x\}$ and $\{v^x, {\tilde{v}}^x , {\tilde{w}}^x, {\tilde{t}}^x\}$, and observe that in at least one of them the (unique) deleted edge is $v^x t^x$ or $v^x {\tilde{t}}^x$. This implies that $F\cap \{v^x u,v^x w^x\}=\emptyset$ or $F\cap \{ v^x {\tilde{v}}^x,v^x {\tilde{w}}^x\}=\emptyset$.
However, in both cases the two undeleted edges ($\{v^x u,v^x w^x\}$ or $\{ v^x {\tilde{v}}^x,v^x {\tilde{w}}^x\}$) form a claw both with $v^x t^x$ and with $v^x {\tilde{t}}^x$,
and therefore $F$ has to delete both $v^x t^x$ and $v^x {\tilde{t}}^x$, for every $x\in \cV(c)$.
But as $|F\setminus \{t^x {\tilde{t}}^x \mid x\in \cV(c)\}|=7$, it follows that $F$ deletes at most one edge incident to $u$, leaving a claw with center $u$ in $G_c-F$: the one with leaves $\{v^x \mid x\in \cV(c)\}$ or the one with leaves $\{{\tilde{v}}^x \mid x\in \cV(c)\}$.
\cqed\maybeqed\end{proof}

\iflncs{\input{figureGadgets-wg}}{\begin{figure}[t]
\centering
\begin{tikzpicture}
\begin{scope}[yscale=-1,rotate=-60]
\node[v] (u)   at (0:0) {$u$};
\node[v] (v1)  at (0:2) {$v^x$};
\node[v] (v1p) at (60:2) {${\tilde{v}}^x$};
\node[v] (v2)  at (120:2) {$v^y$};
\node[v] (v2p) at (180:2) {${\tilde{v}}^y$};
\node[v] (v3)  at (-120:2) {$v^z$};
\node[v] (v3p) at (-60:2) {${\tilde{v}}^z$};
\draw (u)--(v1)--(v1p)--(u)--(v2)--(v2p)--(u)--(v3)--(v3p)--(u);

\begin{scope}[shift={(0:2)}]
	\node[v] (w1)  at (-60:1)   {$w^x$};
	\node[v] (w1p) at (-120:1) {${\tilde{w}}^x$};
	\draw (v1)--(w1)--(w1p)--(v1);
	\node[v] (t1)  at (45:2)  {$t^x$};
	\node[v] (t1p) at (15:2) {${\tilde{t}}^x$};
	\draw (v1)--(t1) (t1p)--(v1);
	\draw[E] (t1)--(t1p);
\end{scope}

\begin{scope}[shift={(120:2)}]
	\node[v] (w2)  at (60:1) {$w^y$};
	\node[v] (w2p) at (0:1) {${\tilde{w}}^y$};
	\draw (v2)--(w2)--(w2p)--(v2);
	\node[v] (t2)  at (165:2)  {$t^y$};
	\node[v] (t2p) at (135:2) {${\tilde{t}}^y$};
	\draw (v2)--(t2)  (t2p)--(v2);
	\draw[E] (t2)--(t2p);
\end{scope}

\begin{scope}[shift={(-120:2)}]
	\node[v] (w3)  at (180:1)  {$w^z$};
	\node[v] (w3p) at (120:1) {${\tilde{w}}^z$};
	\draw (v3)--(w3)--(w3p)--(v3);
	\node[v] (t3)  at (-75:2)   {$t^z$};
	\node[v] (t3p) at (-105:2) {${\tilde{t}}^z$};
	\draw (v3)--(t3)  (t3p)--(v3);
	\draw[E] (t3)--(t3p);
\end{scope}

\end{scope}

\begin{scope}[shift={(9,0)}]
	\node[v] (t) at (180:3) {$t_\top$};
	\node[v] (s) at (180:2) {$s_\top$};
	\draw (t)--(s);
	\node[v] (tp) at (0:3) {$t_\bot$};
	\node[v] (sp) at (0:2) {$s_\bot$};
	\draw (tp)--(sp);
	\node[v] (t1) at (155:3) {$t_{c_1}$};
	\node[v] (s1) at (155:2) {$s_{c_1}$};
	\draw (t1)--(s1);
	\node[v] (t1p) at (130:3) {${\tilde{t}}_{c_1}$};
	\node[v] (s1p) at (130:2) {${\tilde{s}}_{c_1}$};
	\draw (t1p)--(s1p);
	\node[v] (t2) at (102.5:3) {$t_{c_2}$};
	\node[v] (s2) at (102.5:2) {$s_{c_2}$};
	\draw (t2)--(s2);
	\node[v] (t2p) at (77.5:3) {$t_{c_2}$};
	\node[v] (s2p) at (77.5:2) {${\tilde{s}}_{c_2}$};
	\draw (t2p)--(s2p);
	\node[v] (t3) at (50:3) {$t_{c_3}$};
	\node[v] (s3) at (50:2) {$s_{c_3}$};
	\draw (t3)--(s3);
	\node[v] (t3p) at (25:3) {${\tilde{t}}_{c_3}$};
	\node[v] (s3p) at (25:2) {${\tilde{s}}_{c_3}$};
	\draw (t3p)--(s3p);
	\draw (t)--(t1)  (t1p)--(t2)  (t2p)--(t3)  (t3p)--(tp);
	\draw [E] (t1)--(t1p) (t2)--(t2p) (t3)--(t3p);

	\node[v] (t4) at (-36:3) {$t_{d_1}$};
	\node[v] (s4) at (-36:2) {$s_{d_1}$};
	\draw (t4)--(s4);
	\node[v] (t4p) at (-72:3) {${\tilde{t}}_{d_1}$};
	\node[v] (s4p) at (-72:2) {${\tilde{s}}_{d_1}$};
	\draw (t4p)--(s4p);
	\node[v] (t5) at (-108:3) {$t_{d_2}$};
	\node[v] (s5) at (-108:2) {$s_{d_2}$};
	\draw (t5)--(s5);
	\node[v] (t5p) at (-144:3) {${\tilde{t}}_{d_2}$};
	\node[v] (s5p) at (-144:2) {${\tilde{s}}_{d_2}$};
	\draw (t5p)--(s5p);
	\draw (tp)--(t4) (t4p)--(t5) (t5p)--(t);
	\draw [E] (t4)--(t4p) (t5)--(t5p);
\end{scope}

\end{tikzpicture}
\caption{The clause gadget $G_c$ (left) and variable gadget $G^x$ (right) used in the reduction.}
\label{fig:gadgets}
\end{figure}
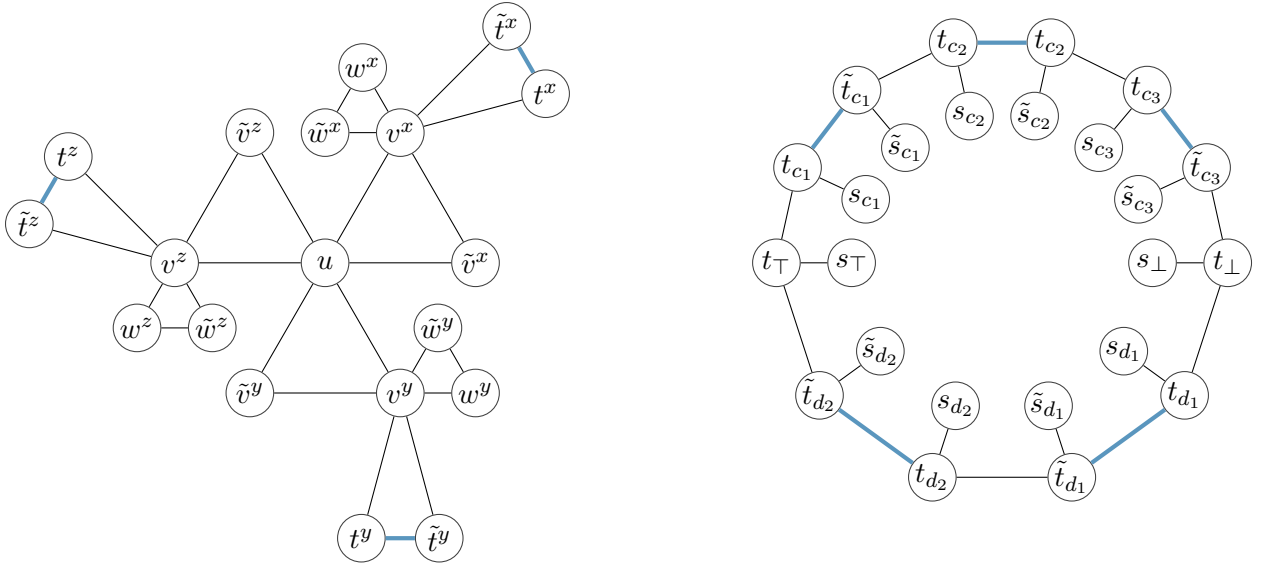

}

Let us now define the variable gadget, as a part of the final graph $G$ constructed for $\phi$.
Take $G$ to be the disjoint union of $G_c$ for all clauses $c$ of $\phi$.
For each variable $x$ of $\phi$, 
add four vertices $t_{\top}^x,t_{\bot}^x,s_{\top}^x,s_{\bot}^x$ to $G$
and then for each $t$-vertex (with tilde or not) in a clause gadget add a corresponding $s$-vertex to $G$.
That is,
\iflncs{%
\begin{align*}
V(G) &= \{u_c \mid c\in \phi\} \cup \{v_c^x,{\tilde{v}}_c^x,w_c^x,{\tilde{w}}_c^x,t_c^x,{\tilde{t}}_c^x,s_c^x,{\tilde{s}}_c^x \mid x\in \cV(c), c\in\phi\} \\
     &\qquad \cup \{t_\top^x,t_\bot^x,s_\top^x,s_\bot^x \mid x\in\cV(\phi)\}.
\end{align*}%
}{%
$$V(G) = \{u_c \mid c\in \phi\} \cup \{v_c^x,{\tilde{v}}_c^x,w_c^x,{\tilde{w}}_c^x,t_c^x,{\tilde{t}}_c^x,s_c^x,{\tilde{s}}_c^x \mid x\in \cV(c), c\in\phi\}
\cup \{t_\top^x,t_\bot^x,s_\top^x,s_\bot^x \mid x\in\cV(\phi)\}.$$}

Let $c_1,c_2,\dots,c_{p(x)}$ be the clauses in which $x$ occurs positively, and
let $d_1,d_2,\dots,d_{q(x)}$ be the clauses in which $x$ occurs negatively.
To the edges defined before (between vertices with the same subscript $c$) we add the following ones (between vertices with the same superscript $x$):
an edge between every $t$-vertex and its $s$-counterpart and
an edge between every two consecutive vertices in the following cycle of length $1+2p(x)+1+2q(x)$:
\iflncs{%
\begin{align*}
&t_\top^x,\quad  t^x_{c_1},{\tilde{t}}^x_{c_1},\ t^x_{c_2},{\tilde{t}}^x_{c_2},\ \dots,\ t^x_{c_{p(x)}},{\tilde{t}}^x_{c_{p(x)}},\\
&\qquad t_\bot^x,\quad t^x_{d_1},{\tilde{t}}^x_{d_1},\ t^x_{d_2},{\tilde{t}}^x_{d_2},\ \dots,\ t^x_{d_{q(x)}},{\tilde{t}}^x_{d_{q(x)}},\quad t_\top^x.
\end{align*}%
}{%
$$t_\top^x,\quad  t^x_{c_1},{\tilde{t}}^x_{c_1},\ t^x_{c_2},{\tilde{t}}^x_{c_2},\ \dots,\ t^x_{c_{p(x)}},{\tilde{t}}^x_{c_{p(x)}},\quad t_\bot^x,\quad t^x_{d_1},{\tilde{t}}^x_{d_1},\ t^x_{d_2},{\tilde{t}}^x_{d_2},\ \dots,\ t^x_{d_{q(x)}},{\tilde{t}}^x_{d_{q(x)}},\quad t_\top^x.$$}
The variable gadget $G^x$ is the subgraph of $G$ induced by the $2(1+2p(x)+1+2q(x))$ vertices named $t$ or $s$ with superscript $x$ (so $G^x$ is a cycle on the $t$-vertices, each with a pendant $s$-vertex attached).

Define $E_\bot^x, E_\top^x$ to be the set of even and odd edges on the above cycle, respectively (so that for all $i$, $t^x_{c_i} {\tilde{t}}^x_{c_i}\in E_\bot^x$ and $t^x_{d_i} {\tilde{t}}^x_{d_i} \in E_\top^x$).
The crucial property is that for any clause $c$ containing $x$, assigning $x\mapsto b$ satisfies the clause $c$ if and only if $t_c^x {\tilde{t}}_c^x$ is \emph{not} in $E_b^x$.
We show that in a minimum HDS exactly one of $E_\bot^x,E_\top^x$ is deleted.

\begin{myclaim}
\label{claim:variable_gadget}
Let $F\subseteq E(G^x)$ be an HDS of $G^x$.
Then $F$ deletes at least $p(x)+q(x)+1$ edges of $G^x$.
Furthermore, if it deletes exactly that many,
then either $F = E_\bot^x$ or $F = E_\top^x$.
\end{myclaim}
\begin{proof}
To show that $F$ deletes at least $p(x)+q(x)+1$ edges of $G^x$ it suffices to notice that each of the $2p(x)+2q(x)+2$ vertices of the cycle is a center of a different claw and each edge deletion hits at most two such claws.

If $F$ deletes exactly that many edges in $G^x$, then each edge of $F$ must be contained in exactly two such claws (hence only edges of the cycle get deleted), and no claw may be hit by two edges (hence no two incident edges get deleted). This means either exactly the even edges $E_\bot^x$ or exactly the odd edges $E_\top^x$ of the cycle are deleted by $F$.
\cqed\maybeqed\end{proof}

We let $k=7m + \sum_{x\in\cV(\phi)}(p(x)+q(x)+1) = 7 m + 3 m + n$, and it is straightforward to verify that $G$ is $\{K_4,\mbox{diamond}\}$-free and $\Delta(G)=6$. Thus, the following lemma encapsulates the final check needed to conclude the proof of Theorem~\ref{thm:claw_ETH}.
\begin{lemma}
$G$ has an HDS of size at most $k$ if and only if $\phi$ is satisfiable.
\end{lemma}
\begin{proof}
(Left to Right) Since $G_c, G^x$ are induced subgraphs of $G$ and the edge sets $E(G^x)$ for $x\in\cV(\phi)$ and $ E(G_c)\setminus\{t_c^x {\tilde{t}}_c^x \mid x\in \cV(c)\}$ for $c\in\phi$ are pairwise disjoint, Claims~\ref{claim:clause_gadget} and~\ref{claim:variable_gadget} imply than any HDS of $G$ has size at least $k$.
Furthermore, equality holds only if it holds in both claims.
Let $F$ be an HDS of size equal to $k$.
From Claim~\ref{claim:variable_gadget}, we infer that $F\cap E(G^x) = E_{b(x)}^x$ for some $b: \cV(\phi) \to \{\bot,\top\}$.
From Claim~\ref{claim:clause_gadget}, for each clause $c$ of $\phi$, 
there is a variable $x$ in $c$ such that $t_c^x {\tilde{t}}_c^x$ is not deleted.
This means $F\cap E(G^x)=E_{b(x)}^x$ does not contain this edge, so by construction we infer that assigning $x \mapsto b(x)$ satisfies clause $c$.
Therefore, each clause is satisfied by assignment $b$.

\medskip

\noindent(Right to Left)
Let $b : \cV(\phi) \to \{\bot,\top\}$ be a satisfying assignment for $\phi$.
We show that by deleting 
$\bigcup_{x\in\cV(\phi)} E^x_{b(x)}$ and 7 more edges in each clause gadgets we can get a claw-free graph.
For each clause $c\in\phi$, there is a variable $y$ such that assigning $y\mapsto b(y)$ satisfies $c$, which by the construction means that $t_c^{y} {\tilde{t}}_c^{y} \not\in E^{y}_{b(y)}$.
Let $F_c = \{u_c v_c^{y}, u_c {\tilde{v}}_c^{y}, v_c^{y} {\tilde{v}}_c^{y}\} \cup \{v_c^x t_c^x, v_c^x {\tilde{t}}_c^x \mid x\in\cV(c), x\neq y\}$.
We claim the set $F:=\bigcup_{x\in\cV(\phi)} E^x_{b(x)} \cup \bigcup_{c\in\phi} F_c$ is an HDS of $G$; note that we have that $|F|=k$.

Observe that if $t_c^x {\tilde{t}}_c^x$ gets deleted by $F$ for some $c\in\phi,x\in\cV(c)$, then $t_c^x {\tilde{t}}_c^x\in E^x_{b(x)}$. Hence, by the construction we infer that assigning $x\mapsto b(x)$ does not satisfy $c$,
which by the definition of $F_c$ implies that edges $v_c^x t_c^x, v_c^x {\tilde{t}}_c^x$ do not get deleted by $F$.

Suppose now that $G-F$ contains a claw.
The center of this claw has degree at least 3 in $G$, so it is a $u$-, $v$-, $t$-, or $\tilde{t}$-vertex.
It cannot be a $u$-vertex, since the closed neighborhood of $u_c$ for any $c\in\phi$ forms three triangles in $G$, exactly one of whose edges are deleted, completely, by $F_c$.
So the closed neighborhood of $u_c$ in $G-F$ forms two triangles joined at $u_c$.
The center also cannot be a $v$-vertex, since the closed neighborhood of $v_c^x$ for any $c\in\phi, x\in\cV(c)$ forms three triangles $\{v_c^x, u_c, {\tilde{v}}_c^x\}$, $\{v_c^x, w_c^x, {\tilde{w}}_c^x\}$ and $\{v_c^x, t_c^x, {\tilde{t}}_c^x\}$ in $G$, for which either the first one gets completely deleted by $F_c$, or $v_c^x t_c^x$ and $v_c^x {\tilde{t}}_c^x$ get deleted by $F_c$. Additionally, $t_c^x {\tilde{t}}_c^x$ may be deleted by $E^x_{b(x)}$, but by the observation of the previous paragraph, this only occurs in the second case (i.e., when $v_c^x t_c^x$ and $v_c^x {\tilde{t}}_c^x$ get deleted by $F_c$).
So the closed neighborhood of $v_c^x$ in $G-F$ again forms two triangles joined at $v_c^x$.
Finally, the center cannot be $t_\bot^x$ or $t_\top^x$ for some $x\in \cV(\phi)$, since these vertices have degree $2$ in $G-F$.

Hence, the center of the claw in $G-F$ must be $t_c^x$ or ${\tilde{t}}_c^x$ for some $c\in\phi, x\in\cV(c)$.
Suppose without loss of generality that it is $t_c^x$.
There are four neighbors of $t_c^x$ in $G$: $s_c^x$, ${\tilde{t}}_c^x$, $v_c^x$, and either ${\tilde{t}}_d^x$ for some $d\in \phi$ or $t_r^x$ for $r\in\{\bot,\top\}$.
Either the edge to ${\tilde{t}}_c^x$ or the edge to ${\tilde{t}}_d^x$ (resp. $t_r^x$) gets deleted by $E^x_{b(x)}$, so for three edges to remain, no others can be deleted. In particular, the edge to $v_c^x$ cannot be deleted and, by the definition of $F$, this occurs only if $v_c^x {\tilde{t}}_c^x$ is also not deleted. By the above observation, the edge to ${\tilde{t}}_c^x$ cannot be deleted.
Hence, only $t_c^x {\tilde{t}}_d^x$ (resp. $t_c^x t_r^x$) gets deleted by $F$ in the neighborhood of $t_c^x$ in $G$, so this neighborhood forms in $G-F$ a triangle with a pendant vertex. 
We have obtained a contradiction in all the cases, so $F$ is indeed an HDS of $G$.
\maybeqed\end{proof}

\section{Conclusions}\label{sec:conc}
In this paper we have charted the parameterized and kernelization complexity of \cdedgedeletion by proving that (i) the problem admits a polynomial kernel, and (ii) the simple $5^k\cdot n^{\Oh(1)}$ branching algorithm following from the observation of Cai~\cite{cai1996fixed} cannot be improved to a subexponential parameterized algorithm, unless the ETH fails.

It should not be a surprise for the reader that the results of this paper were obtained while working on kernelization for {\sc{Claw-free Edge Deletion}}. In this problem, by applying the same vertex modulator principle we arrive at the situation where we have a modulator $X\subseteq V(G)$ with $|X|\leq 4k$, and $G-X$ is a claw-free graph. Then, one can use the structural theorem of Chudnovsky and Seymour~\cite{ChudnovskyS08c,ChudnovskyS08e} (see also variants suited for algorithmic applications, e.g.,~\cite{HermelinML14}) to understand the structure of $G-X$ and of the adjacencies between $X$ and $G-X$. In essence, the structural theorem yields a decomposition of $G-X$ into \emph{strips}, where each strip induces a graph from one of several basic graph classes; each strip has at most two distinguished cliques (possibly equal) called \emph{ends}, and strips are joined together by creating full adjacencies between disjoint sets of ends. Thus, the whole decomposition looks like a line graph, where every vertex is replaced by a possibly larger strip; indeed, the degenerate case where all the strips are single vertices exactly corresponds to the case of line graphs. As far as base classes are concerned, probably the ones most important for understanding the whole decomposition are proper interval graphs and graphs with independent sets of size at most 2 or 3, in particular, co-bipartite graphs. Thus, we believe that for the sake of showing a polynomial kernel for {\sc{Claw-free Edge Deletion}}, one needs to understand the three special cases when $G-X$ is (a) a line graph, (b) a proper interval graph, and (c) a co-bipartite graph.

We believe that the results of this paper present a progress towards this goal by providing a toolbox useful for tackling case (a). In our proof we have used in several places the fact that we exclude also diamonds. However, much of the structural analysis can translated also to the case when only claws are forbidden, so we hope that similar ideas can be also used for understanding case (a), and consequently how the whole decomposition structure should be dealt with in a polynomial kernel for {\sc{Claw-free Edge Deletion}}. Unfortunately, we are currently unable to make any significant progress in cases (b) and (c), of which case (c) seems particularly difficult. 

From another perspective, our positive result gives high hopes for the existence of a polynomial kernel for {\sc{Line Graph Edge Deletion}}, which seems much closer to the topic of this work than {\sc{Claw-free Edge Deletion}}. The problem is that \clawdiamond-free graphs, or equivalently line graphs of triangle-free graphs, have much nicer structural properties than general line graphs. These properties, encapsulated in Lemma~\ref{lem:bag-decomposition}, were used several times to simplify the analysis, which would become much more complicated in the case of general line graphs. Also, note that in this paper the considered graph class can be characterized using only two relatively simple forbidden induced subgraphs. In the case of general line graphs, the classic characterization via forbidden induced subgraphs of Beineke~\cite{Beineke70} involves $9$ different obstacles with up to $6$ vertices.

\bibliographystyle{abbrv}
\bibliography{claw-diamond-deletion}

\begin{thebibliography}{10}

\bibitem{Beineke70}
L.~W. Beineke.
\newblock Characterizations of derived graphs.
\newblock {\em Journal of Combinatorial Theory}, 9(2):129--135, 1970.

\bibitem{bliznets2014interval}
I.~Bliznets, F.~V. Fomin, M.~Pilipczuk, and M.~Pilipczuk.
\newblock A subexponential parameterized algorithm for {I}nterval {C}ompletion.
\newblock {\em CoRR}, abs/1402.3473, 2014.

\bibitem{bliznets2014proper}
I.~Bliznets, F.~V. Fomin, M.~Pilipczuk, and M.~Pilipczuk.
\newblock A subexponential parameterized algorithm for {P}roper {I}nterval
  {C}ompletion.
\newblock In {\em ESA 2014}, volume 8737 of {\em LNCS}, pages 173--184.
  Springer, 2014.

\bibitem{cai1996fixed}
L.~Cai.
\newblock Fixed-parameter tractability of graph modification problems for
  hereditary properties.
\newblock {\em Information Processing Letters}, 58(4):171--176, 1996.

\bibitem{cai2013incompressibility}
L.~Cai and Y.~Cai.
\newblock Incompressibility of {$H$}-free edge modification.
\newblock In {\em IPEC 2013}, volume 8246 of {\em LNCS}, pages 84--96.
  Springer, 2013.

\bibitem{cai2012master}
Y.~Cai.
\newblock Polynomial kernelisation of {$H$}-free edge modification problems.
\newblock Master's thesis, The Chinese University of Hong Kong, Hong Kong,
  2012.

\bibitem{ChudnovskyS08c}
M.~Chudnovsky and P.~D. Seymour.
\newblock Claw-free graphs. {IV.} {D}ecomposition theorem.
\newblock {\em J. Comb. Theory, Ser. {B}}, 98(5):839--938, 2008.

\bibitem{ChudnovskyS08e}
M.~Chudnovsky and P.~D. Seymour.
\newblock Claw-free graphs. {V.} {G}lobal structure.
\newblock {\em J. Comb. Theory, Ser. {B}}, 98(6):1373--1410, 2008.

\bibitem{worker-opl}
M.~Cygan, L.~Kowalik, and M.~Pilipczuk.
\newblock Open problems from workshop on kernels, 2013.
\newblock Available at
  \url{http://worker2013.mimuw.edu.pl/slides/worker-opl.pdf}.

\bibitem{DowneyF99}
R.~G. Downey and M.~R. Fellows.
\newblock {\em Parameterized Complexity}.
\newblock Monographs in Computer Science. Springer, 1999.

\bibitem{drange2014exploring}
P.~G. Drange, F.~V. Fomin, M.~Pilipczuk, and Y.~Villanger.
\newblock Exploring subexponential parameterized complexity of completion
  problems.
\newblock In {\em STACS 2014}, volume~25 of {\em LIPIcs}, pages 288--299.
  Schloss Dagstuhl - Leibniz-Zentrum f\"ur Informatik, 2014.

\bibitem{DrangeP14}
P.~G. Drange and M.~Pilipczuk.
\newblock A polynomial kernel for {T}rivially {P}erfect {E}diting.
\newblock {\em CoRR}, abs/1412.7558, 2014.

\bibitem{flum2006parameterized}
J.~Flum and M.~Grohe.
\newblock {\em Parameterized Complexity Theory}.
\newblock Texts in Theoretical Computer Science. An EATCS Series. Springer
  Berlin Heidelberg, 2006.

\bibitem{FominSV13}
F.~V. Fomin, S.~Saurabh, and Y.~Villanger.
\newblock A polynomial kernel for proper interval vertex deletion.
\newblock {\em {SIAM} J. Discrete Math.}, 27(4):1964--1976, 2013.

\bibitem{fomin2012subexponential}
F.~V. Fomin and Y.~Villanger.
\newblock Subexponential parameterized algorithm for minimum fill-in.
\newblock {\em {SIAM} Journal on Computing}, 42(6):2197--2216, 2013.

\bibitem{ghosh2013faster}
E.~Ghosh, S.~Kolay, M.~Kumar, P.~Misra, F.~Panolan, A.~Rai, and M.~Ramanujan.
\newblock Faster parameterized algorithms for deletion to split graphs.
\newblock {\em Algorithmica}, 2013.

\bibitem{guillemot2013non}
S.~Guillemot, F.~Havet, C.~Paul, and A.~Perez.
\newblock On the \mbox{(non-)}existence of polynomial kernels for {$P_l$}-free
  edge modification problems.
\newblock {\em Algorithmica}, 65(4):900--926, 2013.

\bibitem{HermelinML14}
D.~Hermelin, M.~Mnich, and E.~J. van Leeuwen.
\newblock Parameterized complexity of induced graph matching on claw-free
  graphs.
\newblock {\em Algorithmica}, 70(3):513--560, 2014.

\bibitem{ImpagliazzoPZ01}
R.~Impagliazzo, R.~Paturi, and F.~Zane.
\newblock Which problems have strongly exponential complexity?
\newblock {\em J. Comput. Syst. Sci.}, 63(4):512--530, 2001.

\bibitem{KloksKM94}
T.~Kloks, D.~Kratsch, and H.~M{\"{u}}ller.
\newblock Dominoes.
\newblock In E.~W. Mayr, G.~Schmidt, and G.~Tinhofer, editors, {\em WG 1994},
  volume 903 of {\em LNCS}, pages 106--120. Springer, 1994.

\bibitem{komusiewicz2012cluster}
C.~Komusiewicz and J.~Uhlmann.
\newblock Cluster editing with locally bounded modifications.
\newblock {\em Discrete Applied Mathematics}, 160(15):2259--2270, 2012.

\bibitem{kratsch2009two}
S.~Kratsch and M.~Wahlstr{\"o}m.
\newblock Two edge modification problems without polynomial kernels.
\newblock In {\em IWPEC 2009}, volume 5917 of {\em LNCS}, pages 264--275.
  Springer, 2009.

\bibitem{MetelskyT03}
Y.~Metelsky and R.~Tyshkevich.
\newblock Line graphs of {H}elly hypergraphs.
\newblock {\em {SIAM} J. Discrete Math.}, 16(3):438--448, 2003.

\bibitem{Yannakakis81}
M.~Yannakakis.
\newblock Edge-deletion problems.
\newblock {\em {SIAM} J. Comput.}, 10(2):297--309, 1981.

\end{thebibliography}

\end{document}